\documentclass[11pt]{article}
\usepackage[margin=1.2in]{geometry}
\usepackage[utf8]{inputenc}
\usepackage{tcs-macros}
\usepackage{comment,bbm}
\usepackage[english]{babel}
\usepackage[boxed, linesnumbered]{algorithm2e}

\newcommand{\wt}{\mu}

\newcommand{\poly}{\mathrm{poly}}
\newcommand{\samp}{\mathrm{Samp}}

\newcommand{\sout}[1][]{s_{\mathrm{out}#1}}
\newcommand{\si}[1][]{s_{\mathrm{in}#1}}
\newcommand{\Si}{S_{\mathrm{in}}}

\newcommand{\ninf}{\norm}
\newcommand{\nmax}{\norm[\mathrm{max}]}
\newtheorem{THM}{Theorem}
\crefname{THM}{Theorem}{Theorems}

\title{Optimal Error Pseudodistributions for Read-Once Branching Programs}
\author{Eshan Chattopadhyay\thanks{Supported by NSF grant CCF-1849899.}\\
Cornell University\\
\texttt{eshanc@cornell.edu}
\and
Jyun-Jie Liao\footnotemark[1]\\
Cornell University\\
\texttt{jjliao@cs.cornell.edu}
}
\date{\today}

\begin{document}
\maketitle
\begin{abstract}
In a seminal work, Nisan (Combinatorica'92) constructed a pseudorandom generator for length $n$ and width $w$ read-once branching programs  with seed length $O(\log n\cdot \log(nw)+\log n\cdot\log(1/\varepsilon))$ and  error $\varepsilon$. It remains a central question  to reduce the seed length to $O(\log (nw/\varepsilon))$, which would prove that $\mathbf{BPL}=\mathbf{L}$. However, there has been no improvement on Nisan's construction for the case $n=w$, which is most relevant to space-bounded derandomization.

Recently, in a beautiful work, Braverman, Cohen and Garg (STOC'18) introduced the notion of a \emph{pseudorandom pseudo-distribution} (PRPD) and gave an explicit construction of a PRPD with seed length $\tilde{O}(\log n\cdot \log(nw)+\log(1/\varepsilon))$. A PRPD is a relaxation of a pseudorandom generator, which suffices for derandomizing $\mathbf{BPL}$ and also implies a hitting set. Unfortunately, their construction is quite involved and complicated. Hoza and Zuckerman (FOCS'18) later constructed a much simpler hitting set generator with seed length $O(\log n\cdot \log(nw)+\log(1/\varepsilon))$, but their techniques are restricted to hitting sets.

In this work, we construct a PRPD with seed length 
$$O(\log n\cdot \log (nw)\cdot \log\log(nw)+\log(1/\varepsilon)).$$
This improves upon the construction in \cite{BCG18} by a $O(\log\log(1/\varepsilon))$ factor, and is optimal in the small error regime. In addition, we believe our construction and analysis to be   simpler than the work of Braverman, Cohen and Garg.
\end{abstract}

\section{Introduction}

A major challenge in computational complexity is to understand to what extent randomness is useful for efficient computation. It is widely believed that randomness does not provide substantial savings in time and space for algorithms. Indeed, under plausible assumption, every randomized algorithm for decision problem can be made deterministic with only a polynomial factor slowdown in time ($\BPP=\P$)~\cite{IW97} or a constant factor blowup in space ($\BPL=\L$)~\cite{KvM02}. 

However, it remains open for decades to prove these results unconditionally. For   derandomization in the time-bounded setting, it is known that proving $\BPP=\P$ implies circuit lower bounds which seem much beyond reach with current proof techniques~\cite{KI04}. However no such implications are known for the space-bounded setting, and there has been some progress.   Savitch's theorem~\cite{Sav70} implies that $\RL\subseteq \L^2$. Borodin, Cook, Pippenger~\cite{BCP83} and Jung~\cite{Jung81} proved that $\PL\subseteq \L^2$, which implies $\BPL\subseteq \L^2$. Nisan~\cite{Nis92,Nis94} constructed a pseudorandom generator for log-space computation with seed length $O(\log^2 n)$, and used it to show that $\BPL$ can be simulated with $O(\log^2 n)$ space and polynomial time. Saks and Zhou~\cite{SZ99} used Nisan's generator in a non-trivial way to show that $\BPL\subseteq \L^{3/2}$, which remains the best known result so far. We refer the interested reader to the beautiful survey by Saks \cite{saks1996randomization} for more background and relevant prior work.

We introduce the notion of a read-once branching programs, which is a non-uniform model for capturing algorithms that use limited memory.
\begin{definition}[Read-once branching program]
A $(n,w)$-read-once branching program (ROBP) $B$ is a directed graph on the vertex set $V=\bigcup_{i=0}^n V_i$, where each set $V_i$ contains $w$ nodes. Every edge in this directed graph is labeled either $0$ or $1$. For every $i<n$, and every node $v\in V_i$, there exists exactly two edges starting from $v$, one with label $0$ and the other with label $1$. Every edge starting from a node in $V_i$ connects to a node in $V_{i+1}$.
We say $n$ is the \emph{length} of $B$, $w$ is the \emph{width} of $B$ and $V_i$ is the $i$-th \emph{layer} of $B$.

Moreover, there exists exactly one starting state $s\in V_0$, and exactly one accepting state $t\in V_n$. For every $x=(x_1,\dots,x_n)\in\bits{n}$, we define $B(x)=1$ if starting from $s$ we will reach $t$ following the edges labeled by $x_1,\dots,x_n$. Otherwise we define $B(x)=0$.
\end{definition}
It is well-known the computation of a probabilistic Turing machine that uses space $S$ and tosses $n$ coins,  on a given input $y$, can be carried out by a $(n,2^{O(S)})$-ROBP $B_y$. In particular, if the string $x \in \zo^n$ corresponds to the $n$ coin tosses, then  $B_y(x)$ is the output of the Turing machine. 

A standard derandomization technique is via \emph{pseudorandom generators}. We define this notion for the class of ROBPs.

\begin{definition}[Pseudorandom generator]
A function $G:\bits{s}\to\bits{n}$ is a $(n,w,\eps)$-pseudorandom generator (PRG) if for every $(n,w)$-ROBP $B$,
$$\abs{\ex[x\in\bits{n}]{B(x)}-\ex[r\in\bits{s}]{B(G(r))}}\le\eps.$$
The \emph{seed length} of $G$ is $s$.   $G$ is explicit if $G$ is computable in $O(s)$ space.
\end{definition}
To derandomimze  space-bounded computation  given an explicit $(n,w,\eps)$-PRG, one can enumerate $B(G(r))$ for every $r\in\bits{s}$ with $O(s)$ additional space to compute an $\eps$-approximation of the quantity $\ex[x]{B(x)}$. 

Nisan~\cite{Nis92} constructed a $(n,w,\eps)$-PRG with seed length $O(\log n \cdot \log (nw/\eps))$, which implies $\BPL\subseteq \L^2$.   While there is a lot of progress in constructing PRG with better seed length for restricted family of ROBP (see, e.g., \cite{NZ96,Arm98,BV10,BDVY13,BRRY14,KNP11,De11,Ste12,MRT19} and references therein), Nisan's generator and its variants~\cite{Nis92,INW94,RR99} remain the best-known generators in the general case.
\subsection{Pseudorandom pseudodistribution }
Recently, a beautiful work of Braverman, Cohen and Garg \cite{BCG18} introduced the notion of a \emph{pseudorandom pseudodistribution} (PRPD)  that relaxes the definition of a PRG.
\begin{definition}[Pseudorandom pseudodistribution]
A pair of functions $(G,\rho):\bits{s}\to \bits{n}\times \bbR$ generates a $(n,w,\eps)$-pseudorandom pseudodistribution (PRPD) if for every $(n,w)$-ROBP $B$,
$$\abs{\ex[x\in\bits{n}]{B(x)}-\ex[r\in\bits{s}]{\rho(r)\cdot B(G(r))}}\le\eps.$$
We say $s$ is the \emph{seed length} of $(G,\rho)$. We say $(G,\rho)$ is $k$-bounded if $|\rho(x)|\le k$ for every $x\in\bits{s}$. We say $(G,\rho)$ is explicit if they are computable in space $O(s)$.
\end{definition}
Note that a $(n,w,\eps)$-PRG $G$ of seed length $s$ with a constant function $\rho(x)=1$ generates a $1$-bounded $(n,w,\eps)$-PRPD. Similar to a PRG, it is possible to derandomize $\BPL$ by enumerating all seeds of a PRPD and computing an $\eps$-approximation for $\ex[x]{B(x)}$. In \cite{BCG18} they  observe that given $(G,\rho)$ which generates an $(n,w,\eps)$-PRPD, the function $G$ itself is an $\eps$-hitting set generator for $(n,w)$-ROBP. 

The main result in \cite{BCG18} is an explicit construction of a $(n,w,\eps)$-PRPD with seed length
$$O\left(\left(\log n\cdot \log (nw)+\log(1/\eps)\right)\cdot \log\log(nw/\eps)\right),$$
which is $\poly(nw/\eps)$-bounded.\footnote{Note that in \cite{BCG18}, they define $\sum_{r} \rho(r) B(G(r))$ to be the approximation of $\ex[x]{B(x)}$. Here we define $\ex[r]{\rho(r) B(G(r))}$ to be the approximation instead to emphasize the possible loss when plugged into the Saks-Zhou scheme. (See \Cref{sec:Saks-Zhou} for more details.) Therefore a $k$-bounded PRPD in their definition is actually $2^s k$-bounded in our definition. Nevertheless, one can show that their construction is still $\poly(nw/\eps)$-bounded with our definition.} This improves on the seed-length of Nisan's generator and provides near optimal dependence on error. 

Unfortunately, the construction and analysis in \cite{BCG18} is highly complicated.  Hoza and Zuckerman~\cite{HZ18} provided a dramatically simpler hitting set generator   with slightly improved seed length. However, it is not clear how to extend their techniques  for constructing a PRPD (or PRG).

\subsection{Main result}
In this paper, we construct a PRPD with optimal dependence on error (up to  constants).
\begin{THM}
\label{thm:main}
There exists an explicit $(n,w,\eps)$-PRPD generator $(G,\rho)$ with seed length
$$O\left(\log n\cdot \log(nw)\cdot \log\log (nw) + \log(1/\eps)\right),$$
which is $\poly(1/\eps)$-bounded.
\end{THM}
\noindent
This improves upon the construction in \cite{BCG18} by a factor of $O(\log\log(1/\eps))$, for any $\eps<n^{-\Omega(\log(nw)\log\log(nw))}$.  

As observed in \cite{BCG18}, the small-error regime is well motivated for application to derandomizing space-bounded computation. In particular, Saks and Zhou \cite{SZ99} instantiated Nisan's PRG with error $n^{-\omega(1)}$ to obtain the result $\BPL\subseteq \L^{3/2}$. We note that one can replace the PRG in the Saks-Zhou scheme with a PRPD which is $\poly(w,1/\epsilon)$-bounded, and hence  improvements to our result will lead to improved derandomization of $\BPL$. We sketch a proof in \Cref{sec:Saks-Zhou}.

Our construction uses a  strategy similar to \cite{BCG18} with the following key differences. 
\begin{itemize}
\item The construction in \cite{BCG18} has a more \emph{bottom-up} nature:  their construction follows the binary tree structure in Nisan's generator~\cite{Nis92}, but in each node they maintain a sophisticated ``leveled matrix representation" (LMR) which consists of many pieces of small-norm matrices, and they show how to combine pieces in two LMRs one by one to form a LMR in the upper level. Our construction  follows the binary tree structure in Nisan's generator, but has a more \emph{top-down} spirit. We give a clean recursive formula which generates a ``robust PRPD" for $(n,w)$-PRPD given robust PRPDs for $(n/2,w)$-ROBP, where a robust PRPD is a family of pseudodistributions such that the approximation error of pseudodistribution drawn from this family is small on average. (A formal definition can be found in \Cref{def:robust-PRPD}.)
The top-down nature of our construction significantly simplifies the construction and analysis. 

\item Following \cite{BCG18}, we use an averaging sampler in our recursive construction, but we further observe that we can apply a simple ``flattening" operation to limit the growth of seed length. With this observation, we not only improve the seed length but also simplify the construction and analysis by avoiding some special case treatments that are necessary in \cite{BCG18}. (Specifically, we do not need the special multiplication rule ``outer product" in \cite{BCG18}.)
\end{itemize}
\paragraph{Independent work.} Independent work of Cheng and Hoza~\cite{CH20} remarkably prove that a hitting set generator (HSG) for ROBPs can be used for derandomizing $\BPL$. Their first result shows that every $(n,w)$-ROBP $f$ can be deterministically approximated within error $\eps$ with an explicit HSG for $(\poly(\frac{nw}{\eps}),\poly(\frac{nw}{\eps}))$-ROBP with seed length $s$. The space complexity of their first derandomization is $O(s+\log(nw/\eps))$. Their second result shows that every $(n,w)$-ROBP $f$ can be deterministically approximated within error $\eps$ with an explicit HSG for $(n,\poly(w))$-ROBP with seed length $s$. Their second derandomization has space complexity $O(s+w\log(n/\eps))$, and only requires black-box access to $f$. 

Their first result does not imply better derandomization algorithms with the state-of-art HSGs so far. Plugging in the HSG from \cite{HZ18}, their second result gives a black-box derandomization algorithm for $(n,w)$-ROBP in space $O(\log(n)\log(nw)+w\log(n/\eps))$. This is better than the black-box derandomization with our PRPD for the restricted case of $w~=~O(1)$. We note that an advantage of PRPDs (over hitting sets) is that they are applicable in the Saks and Zhou's scheme~\cite{SZ99} (as mentioned in \Cref{sec:Saks-Zhou}, when applied with Armoni's sampler trick~\cite{Arm98}). 

\paragraph{Organization.} In \Cref{sec:matrices}, we present the matrix representation of ROBPs, see how a pseudodistribution can be interpreted as matrices, and introduce some basic rules for translating between matrix operations and operations on pseudodistribution. We use  \Cref{sec:proof-overview} to present an outline of our main construction and proof. \Cref{sec:prelim} contains  necessary preliminaries. In \Cref{sec:samp-our}, we formally prove several lemmas about using samplers on approximate matrix multiplication. In \Cref{sec:main-construction}, we present and prove correctness of our main construction. We conclude with possible future directions in \Cref{sec:open-questions}.

\section{ROBPs and Matrices}\label{sec:matrices}
We introduce the matrix representation of ROBPs and some related definitions that are useful in the rest of the paper. First, we setup some notation.
\newline
\textbf{Notation:} Given two strings $x,y$, we use $x\Vert y$ to denote the concatenation of $x$ and $y$. For every $n\in\bbN$, we use $[n]$ to denote the set $\{1,2,\dots,n\}$. We denote a collection of objects $A_i^j$ with subscript $i\in S$ and superscript $j\in T$ by $[A]_S^T$ for short.

Given a $(n,w)$-ROBP $B$ with layers $V_0,\dots,V_n$, we can represent the transition from layer $V_{t-1}$ to $V_{t}$ by two stochastic matrices $M_t^0$ and $M_t^1$ as follows: suppose layer $V_j$ consists of the nodes $\{v_{j,1},\dots,v_{j,w}\}$.  The entry $(M_t^0)_{i,j}=1$ if and only if there exist a $0$-labeled edge from $v_{t-1,i}$ to $v_{t,j}$ (else $(M_t^0)_{i,j}=0$). The matrix $M_t^1$ is defined similarly according to the edges that labeled $1$ between layers $V_{t-1}$ and $V_t$. More generally, we can also represents multi-step transition by a stochastic matrix. That is, for every $0\le a\le b\le n$, and every $r=(r_{a+1},\ldots,r_b)\in\bits{b-a}$, we can define
$$M_{a..b}^r=\prod_{t=a+1}^b M_t^{r_t}$$
which corresponds to the transition matrix from layer $a$ to layer $b$ following the path labeled by $r$. Note that every row of $M_{a,b}^r$ contains exactly one $1$, and the other entries are $0$.

An $n$-step random walk starting from the first layer can be represented with the following matrix:
$$M_{0..n}=\frac{1}{2^n}\sum_{r\in\bits{n}}M_{0..n}^r=\prod_{t=1}^n \frac{1}{2}\left(M_t^0+M_t^1\right).$$
By definition of $M_t^0,M_t^1$ one can observe that the $(i,j)$ entry of $M_{0..n}$ is the probability that a random walk from $v_{0,i}\in V_0$  reaches $v_{n,j}\in V_n$. Therefore, suppose $v_{0,i}\in V_0$ is the starting state of $B$, $v_{n,j}\in V_n$ is the accepting state of $B$, then $\ex[x]{B(x)}$ equals the $(i,j)$ entry of $M_{0..n}$. 

Recall that a generator of a $(n,w,\eps)$-PRPD is a pair of function $(G,\rho)$ such that for every $(n,w)$-ROBP $B$,
$$\abs{\ex[r]{\rho(r)\cdot B(G(r))}-\ex[x\in\bits{n}]{B(x)}}\le \eps.$$
Equivalently, for every transition matrices $M_1^0,M_1^1,\ldots,M_n^0,M_n^1$, we have 
$$\nmax{\ex[r]{\rho(r)\cdot M_{0..n}^{G(r)}} - M_{0..n}}\le \eps,$$
where $\nmax{A}$ denotes $\max_{i,j}\abs{A(i,j)}$. 

Therefore it is natural to represents a PRPD $(G,\rho)$ with a mapping $\cG:\bits{s}\to \bbR^{w\times w}$ where $\cG(r)=\rho(r)\cdot M^{G(r)}_{0..n}$. More generally, we will use a notation similar to the ``matrix bundle sequence" (MBS) introduced in \cite{BCG18} to represent a PRPD.

\begin{definition}\label{def:MBS}
Consider a $(n,w)$-ROBP $[M]_{[n]}^{\bit}$ and a pair of functions $(G,\rho):\bits{\sout}\times[\Si]\to\bits{n}\times\bbR$. The \emph{matrix form} of $(G,\rho)$ on $[M]_{[n]}^{\bit}$ is a mapping $\cA:\bits{\sout}\times[\Si]\to\bbR^{w\times w}$ such that for every $x\in\bits{\sout}$ and $y\in[\Si]$,
$$\cA(x,y)=\rho(x,y)\cdot M_{0..n}^{G(x,y)}.$$
For every $x\in\bits{\sout}$ we abuse the notation and define
$$\cA(x)=\ex[y]{\cA(x,y)}.$$
Besides, we define $\inp{\cA}=\ex[x,y]{\cA(x,y)}$. We say $\sout$ is the \emph{outer seed length} of $\cA$, denoted by $\sout(\cA)$, and $\Si$ is the \emph{inner size} of $\cA$, denoted by $\Si(\cA)$. We also define $\si(\cA)=\ceil{\log \Si}$ to be the \emph{inner seed length} of $\cA$, and $s(\cA)=\sout(\cA)+\si(\cA)$ to be the \emph{seed length} of $\cA$.
\end{definition}
\begin{remark}
For every fixed $x$, the collection $\{\cA(x,y):y\in[\Si]\}$ corresponds to the ``matrix bundle" in \cite{BCG18}. This should be treated as a collection of matrices which ``realizes" the matrix $\cA(x)$. The whole structure $\cA$ corresponds to the ``matrix bundle sequence" in \cite{BCG18}, and should be treated as a uniform distribution over the set $\{\cA(x):x\in\bits{\sout}\}$. 
\end{remark}

When the ROBP $[M]_{[n]}^{\bit}$ is clear in the context, we will use the matrix form $\cA$ to represent the pseudodistribution $(G,\rho)$ directly. We will apply arithmetic operations on matrices $\cA(x)$, and these operations can be easily translated back to operations on pseudodistributions as follows.
\begin{definition}
Consider a $(n,w)$-ROBP $[M]_{[n]}^{\bit}$, and a pair of function $(F,\sigma):[S]\to \bits{n}\times\bbR$. The matrix that is \emph{realized} by $(F,\sigma)$ on $M_{0..n}$ is $\ex[{i\in[S]}]{\sigma(i)\cdot M_{0..n}^{F(i)}}$. We say S is the \emph{size} of $(F,\sigma)$.
\end{definition}
\noindent
Scaling the matrix corresponds to scaling the coefficients in the pseudodistribution.
\begin{claim}
Consider a $(n,w)$-ROBP $[M]_{[n]}^{\bit}$, let $A$ be a matrix realized by matrix bundle $(F_A,\sigma_A)$ on $M_{0..n}$. Then $cA$ is realized by a matrix bundle $(F'_A,\sigma'_A)$ of size $S_A$ s.t. $F'_A=F_A$ and $\sigma'_A(x)=c\sigma_A(x)$ for every $x\in[S]$.
\end{claim}
\noindent
The summation on matrices corresponds to re-weighting and union on pseudodistributions.
\begin{claim}
Consider a $(n,w)$-ROBP $[M]_{[n]}^{\bit}$, let $A$ be a matrix realized by matrix bundle $(F_A,\sigma_A)$ of size $S_A$ on $M_{0..n}$ and $B$ be a matrix realized by matrix bundle $(F_B,\sigma_B)$ of size $S_B$ on $M_{0..n}$. Then $A+B$ is realized by a matrix bundle $(F',\sigma')$ of size $S_A+S_B$ on $M_{0..n}$ s.t. 
\[
F'(x) = \begin{cases} F_A(x) &\mbox{if } x \le S_A \\
F_B(x-S_A) & \mbox{if }  x > S_A \end{cases}\textrm{ and }
\sigma'(x) = \begin{cases} \frac{S_A+S_B}{S_B}\cdot \sigma_A(x) &\mbox{if } x \le S_A \\
\frac{S_A+S_B}{S_B}\cdot \sigma_B(x-S_A) & \mbox{if }  x > S_A \end{cases}
\]
\end{claim}
\noindent
The multiplication on matrices corresponds to concatenation of pseudodistributions.
\begin{claim}
Consider a $(n,w)$-ROBP $[M]_{[n]}^{\bit}$, let $A$ be a matrix realized by matrix bundle $(F_A,\sigma_A)$ of size $S_A$ on $M_{0..n/2}$ and $B$ be a matrix realized by matrix bundle $(F_B,\sigma_B)$ of size $S_B$ on $M_{n/2..n}$. Fix a bijection $\pi:[S_A]\times[S_B]\to{[S_A\cdot S_B]}$. Then $AB$ is realized by a matrix bundle $(F',\sigma')$ of size $S_A\cdot S_B$ s.t. for every $a\in[S_A],b\in[S_B]$,
$$F'(\pi(a,b))=F_A(a)\Vert F_B(b)\textrm{ and }\sigma'(\pi(a,b))=\sigma(a)\cdot \sigma(b).$$
\end{claim}

\section{Proof Overview}\label{sec:proof-overview}
In this section we give an outline of our construction and proof. In \Cref{sec:samp-arg}, we briefly recap how a sampler is used in \cite{BCG18} to achieve better seed length in the small-error regime. We discuss our construction ideas in \Cref{sec:our-const}. 
\subsection{The sampler argument}\label{sec:samp-arg}
 Nisan's generator and its variants recursively use a lemma of the following form.
\begin{lemma}\label{lemma:INW}
Consider a $(n,w)$-ROBP $[M]_{[n]}^{\bit}$. Let $\cA$ be the matrix form of a distribution on $M_{0..n/2}$, and $\cB$ be the matrix form of a distribution on $M_{n/2..n}$. Suppose $s(\cA)=s(\cB)=s$. Then there exists a distribution whose matrix form $\cC$ on $M_{0..n}$ of seed length $s+O(\log(w/\delta))$ such that
$$\nmax{\inp{\cC}-\inp{\cA}\inp{\cB}}\le \delta.$$
\end{lemma}
This lemma is usually achieved with a pseudorandom object. For example, the INW generator~\cite{INW94} uses a bipartite expander with degree $\poly(w/\delta)$ to construct the distribution $\cC$ in the above lemma. That is, for every edge $(x,y)$ in the expander $G$, they  add $\cA(x)\cB(y)$ into $\cC$. A similar lemma can also be obtained with universal hash functions~\cite{Nis92} or seeded extractors~\cite{RR99}. By recursively constructing good approximations of $M_{0..n/2}$ and $M_{n/2..n}$ and applying \Cref{lemma:INW}, one can obtain a PRG which has seed length $O(\log n\cdot\log(nw/\eps))$ ($\delta$ is taken to be $\eps/n$ because of a union bound). Observe that in such constructions, one needs to pay $O(\log(1/\eps))$ (in seed length) per level of recursion. 

The crucial idea in \cite{BCG18} is to amortize this cost over all $\log n$ levels. What makes this possible is the following argument, which we will refer to as the \emph{sampler argument}. First we define the notion of an averaging sampler. 
\begin{definition}
A function $g:\bits{n}\times\bits{d}\to\bits{m}$ is an $(\eps,\delta)$-(averaging) sampler if for every function $f:\bits{m}\to{[0,1]}$, $$\pr[x\in\bits{n}]{\abs{\ex[s\in\bits{d}]{f(g(x,s))}-\ex[y\in\bits{m}]{f(y)}}\le\eps}\ge 1-\delta.$$
\end{definition}
\noindent
The crucial observation in \cite{BCG18} is that if one uses a sampler to prove \Cref{lemma:INW}, the error actually scales with the norm of one of the matrix forms.
\begin{lemma}[\cite{BCG18}]\label{lemma:samp-arg}
Consider a $(n,w)$-ROBP with matrix representation $[M]_{[n]}^{\bit}$. Let $\cA$ and $\cB$ be (pseudo)distributions in matrix form on $M_{0..n/2}$ and $M_{n/2..n}$ respectively. Let $n=\sout(\cA)$, $m=\sout(\cB)$. Suppose $\forall x\in\bits{n},\norm{\cA(x)}\le 1$ and $\forall y\in\bits{m},\norm{\cB(y)}\le 1$. Let $g:\bits{n}\times \bits{d}\to\bits{m}$ be a $(\eps,\delta)$ sampler. Then there exists a (pseudo)distribution $\cC$ such that 
$$\norm{\inp{C}-\inp{A}\inp{B}}\le O\left(w^2\left(\delta+\eps \ex[x]{\norm{\cA(x)}}\right)\right).$$
Besides, $\cC$ has outer seed length $n=\sout(\cA)$, and for every $x\in\bits{n}$, $$\cC(x)=\ex[s]{\cA(x)\cB\left(g(x,s)\right)}.$$
Note that $\si(\cC)=\si(\cA)+\si(\cB)+d$.
\end{lemma}
The intuition behind this approximation is as follows. If we want to compute the matrix product precisely, we take every $\cA(x)$ and multiply it with $\ex[y]{\cB(y)}$. However, with the help of sampler, we can use $x$ as our seed to select some samples from $\cB$, and take their average as an estimate of $\ex[y]{\cB(y)}$. The error of this approximation comes in two different way. For those $x$ which are not good choices of a seed for the sampler, the samples chosen with such an $x$ can deviate from the average arbitrarily. However, only $\delta$ fraction of $x$ can be bad, so they incur at most $\delta$ error. The second kind of error is the estimation error between average of samples $\ex[s]{\cB(g(x,s))}$ and the real average $\ex[y]{\cB(y)}$, which can be at most $\eps$. Since this gets multiplied  with $\cA(x)$, this kind of error actually scales with $\norm{\cA(x)}$. Although the first kind of error (which is $\delta$) does not benefit from $\norm{\cA}$ being small, in \cite{BCG18} they observe that, the parameter $\delta$ has almost no influence on the seed length in some cases. To discuss this more precisely, we first recall explicit constructions of samplers.
\begin{lemma}[\cite{RVW01,Gol11}]
For every $\delta,\eps>0$ and integer $m$, there exists a space efficient $(\eps,\delta)$-sampler $f:\bits{n}\times\bits{d}\to\bits{m}$ s.t. $d=O(\log\log(1/\delta)+\log(1/\eps))$ and $n=m+O(\log(1/\delta))+O(\log(1/\eps))$.
\end{lemma}
Note that in \Cref{lemma:samp-arg}, $s(\cC)=s(\cA)+d+\si(\cB)$. Therefore if $n\ge m+O(\log(1/\delta))+O(\log(1/\eps))$, $\delta$ has almost no impact on the seed length. 

To use the above ideas, it boils down to working with matrices with small norm, and  making sure that every multiplication is ``unbalanced" enough so that $\delta$ has no impact. \cite{BCG18} applies a delicate telescoping sum trick (which they called ``delta sampler") to divide an $\eps$-approximation into a base approximation with $1/\poly(n)$ error and several ``correcting terms" which have small norm. By carefully choosing the samplers and discarding all the non-necessary terms,  they roughly ensure the following properties: first, a matrix with large seed length must have small norm; second, every matrix multiplication is unbalanced enough so that $\delta$ has no impact on the seed length.

With these properties and the sampler argument, they show that the total seed length is bounded by $\tilde{O}(\log(1/\eps)+\log n \log(nw))$.

\subsection{Our construction}\label{sec:our-const}
In executing the ideas sketched above,  the construction and analysis in \cite{BCG18} turns out to be quite complicated and involved. One thing which complicates the construction and analysis is its \emph{bottom-up} nature. That is, when multiplying two terms, they create more terms with the telescoping sum trick. Moreover, in the telescoping sum trick one needs to choose the parameters of each sampler very carefully to make sure the seed length of each term does not exceed its ``smallness". 

Our first step toward a simpler construction is the following \emph{top-down} formula, which we will apply recursively to compute an approximation of $M_{0..n}$:
\begin{lemma}\label{lemma:main}
Let $\norm{\cdot}$ be a sub-multiplicative matrix norm, and $A,B$ be two matrices s.t. $\norm{A},\norm{B}\le 1$. Let $k\in\bbN$ and $\gamma<1$. For every $0\le i\le k$, let $A_i$ be a $\gamma^{i+1}$-approximation of $A$, and let $B_i$ be a $\gamma^{i+1}$-approximation of $B$. Then $$\sum_{i=0}^k A_{i}B_{k-i}-\sum_{i=0}^{k-1} A_{i}B_{k-1-i}$$ is a $((k+2)\gamma^{k+1}+(k+1)\gamma^{k+2})$-approximation of $AB$.
\end{lemma}
\begin{proof}
We have,
\begin{align*}
&\norm{ (\sum_{i=0}^k A_{i}B_{k-i}-\sum_{i=0}^{k-1} A_{i}B_{k-1-i})-AB}\\
 &= \norm{ \sum_{i=0}^k (A-A_{i})(B-B_{k-i})-\sum_{i=0}^{k-1} (A-A_{i})(B-B_{k-1-i})+(A_k-A)B+A(B_k-B)}\\
 &\le \sum_{i=0}^k \norm{A-A_i}\cdot\norm{B-B_{k-i}}+\sum_{i=0}^{k-1}\norm{A-A_{i}}\cdot\norm{B-B_{k-1-i}}\\
 &+\norm{A_k-A}\cdot\norm{B}+\norm{A}\cdot\norm{B_k-B}\\
 &\le(k+2)\gamma^{k+1}+(k+1)\gamma^{k+2}
\end{align*}
\end{proof}
This formula shares an important property with the BCG construction: we never need a $\gamma^k$-approximation (which implies large seed length) on both sides simultaneously. The benefit of our top-down formula is that we are treating the PRPD as one object instead of the sum of many different terms. One obvious effect of such treatment is we don't need to analyze the ``smallness" of each term and the accuracy of the whole PRPD separately. 

In this top-down formula, we do not explicitly maintain small-norm matrices as in \cite{BCG18}. However, observe that in the proof of \Cref{lemma:main}, we are using the fact that $A_k-A$ is a small norm matrix. Our goal is to apply the sampler argument (\Cref{lemma:samp-arg}) on these ``implicit" small-norm matrices. The following is our main technical lemma.

\begin{lemma}[main lemma, informal]\label{lemma:main-informal}
Let $A,B\in\bbR^{w\times w}$, $k\in\bbN$ and $\gamma<1$. Suppose for every $i\le k$ there exists pseudodistribution $\cA_i,\cB_i$ such that $\ex[x]{\norm{\cA_i(x)-A}}\le \gamma^{i+1}$, $\ex[x]{\norm{\cB_i(x)-B}}\le \gamma^{i+1}$, and $\norm{\cA_i(x)},\norm{\cB_i(x)}\le 1$ for every $x$. Then there exists a pseudo-distribution $\cC_k$ such that $$\ex[x]{\norm{\cC_k(x)-AB}\le O(\gamma)^{k+1}},$$
where $\cC_k(x)=\sum_{i+j=k} A_{x,i}B_{x,j}-\sum_{i+j=k-1} A_{x,i}B_{x,j}$. $A_{x,i}$ and $B_{x,i}$ are defined as follows.
\begin{itemize}
\item 
If $i>\ceil{k/2}$, $A_{x,i}=\cA_i(x)$ and $B_{x,i}=\cB_i(x)$.
\item
If $i\le \ceil{k/2}$, $A_{x,i}=\ex[s]{\overline{\cA_i}(g_i(x,s))}$ and $B_{x,i}=\ex[s]{\overline{\cB_i}(g_i(x,s))}$, where $g_i$ is a $(\gamma^{i+1},\gamma^{k+1})$-sampler, and $\overline{\cA_i},\overline{\cB_i}$ denote the ``flattened" form of $\cA_i$ and $\cB_i$.
\end{itemize}
\end{lemma}

We leave the explanation of ``flattened" for later and explain the intuition behind the lemma first. Our goal is to construct $\cC_k$ such that $\cC_k(x)$ is a good approximation of $AB$ \emph{on average} over $x$. We  know that $\cA_i$ and $\cB_i$ are $\gamma^{i+1}$-approximation of $A$ and $B$ on average. Our hope is to use $x$ to draw samples $A_i$ and $B_i$ from $\cA_i$ and $\cB_i$, and apply the formula in \Cref{lemma:main} to get a good approximation of $AB$. In particular, a natural choice would be setting $A_{x,i}=\cA_i(x)$ and $B_{x,i}=\cB_i(x)$ for every $i\le k$. However, if there exists a term $A_{x,i}B_{x,j}$ such that $A_{x,i}$ and $B_{x,j}$ are both bad approximation for a large enough fraction of $x$, we cannot guarantee to get a $O(\gamma^{k+1})$-approximation on average.

To avoid the above case, for every $i\le\ceil{k/2}$ we use a sampler to approximate $\inp{\cA_i}$ and $\inp{\cB_i}$. This ensure that the chosen samples $A_{x,i}$ and $B_{x,i}$ are good \emph{with high probability}. This guarantees that in each term $A_{x,i}B_{x,j}$, at least one of $A_{x,i}$ or $B_{x,j}$ will be a good choice with high probability over $x$. If $A_{x,i}$ is a good choice with high probability, we can apply the average-case guarantee on $B_{x,i}$ to get an average-case guarantee for $\cC_k$, and vice versa. (Indeed, this is the sampler argument.) Therefore we can ensure that $\cC_k(x)$ is good on average. Note that we only apply a sampler on $\cA_i$ (or $\cB_i$) when $i\le \ceil{k/2}$, which means $\cA_i$ (or $\cB_i$) has small seed length. Therefore we don't need to add too much redundant seed to make the sampler argument work.

In executing the above sketched idea, we run into the following problem: in each multiplication, the inner seed on both sides  aggregates to the upper level. If we start with pseudodistributions with non-zero inner seed in the bottom level, the inner seed would become $\Omega(n)$ in the topmost level. Therefore we need a way to limit the aggregation of inner seed.

In \cite{BCG18}, they run into a similar problem. To deal with this,  they apply a different multiplication rule, ``outer product", in some special cases to deal with this. However, the outer product does not seem applicable in our construction. Nevertheless, we observe that whenever we use a sampler to select matrix $A_{x,i}$, we only care about whether $\inp{\cA_i}$ is close to $A$, and we don't need most of $\cA_i(x)$ to be close to $A$ anymore. Therefore we will ``flatten" $\cA_i$ whenever we apply a sampler. That is, recall that each $\cA_i(x)$ is realized by the average of some matrices, $\ex[y]{\cA_i(x,y)}$. We define the flattened form of $\cA_i$, denoted by $\overline{\cA}_i$, such that $\overline{\cA}_i(x\Vert y)=\cA_i(x,y)$. Observe that $\inp{\overline{\cA_i}}=\inp{\cA_i}$ and $\si(\overline{\cA_i})=0$. This guarantees that the inner seed length of $\cA_i$ will not aggregate in $\cC_k$. Moreover, while the flattening will increase the outer seed length of $\overline{\cA_i}$, this is almost for free since we only flatten $\cA_i$ when $i\le\ceil{k/2}$, i.e. when $\cA_i$ has relatively small seed length. As a result, this operation also helps us save a $O(\log\log(1/\eps))$ factor in the seed length.

We  conclude by briefly discussing the seed length analysis. First note that we set $\gamma=1/\poly(n)$ to make sure that the error is affordable after a union bound. Now consider the inner seed length. Consider a term $A_iB_j$ such that $i\ge j$. In this term, part of the inner seed of $\cC$ is passed to $\cA_i$, and the other is used for the sampler on $B_j$. Since the seed length of the sampler only needs to be as large as the ``precision gap" between $\cA_i$ and $\cC_k$, the inner seed length of $\cC_k$ can be maintained at roughly $O(k\log(1/\gamma))=O(\log(1/\eps))$. However, after each multiplication, there's actually a $O(\log (nw/\gamma))=O(\log (nw))$ additive overhead. Note that this is necessary since the $k=0$ case degenerates to the INW generator. Therefore after $\log n$ levels of recursion, the inner seed length will be $O(\log(1/\eps)+\log n\cdot \log (nw))$. 

Besides, we also need the outer seed length of $\cC_k$ to be long enough so that we can apply a sampler on $\cA_{\ceil{k/2}}$ and $\cB_{\ceil{k/2}}$. The seed length caused by approximation accuracy $\eps$ can be bounded similarly as the inner seed length. However, the $O(\log n\cdot \log (nw))$ inner seed length will be added to the outer seed length several times, because of the flattening operation. Nevertheless, since we only do flattening for $\cA_{i}$ and $\cB_{i}$ where $i\le \ceil{k/2}$, this ensures that the flattening operation happens at most $\log k$ times. So the total outer seed length will be bounded by $O(\log(1/\eps)+\log k\cdot \log n\cdot \log (nw))=O(\log(1/\eps)+\log \log(1/\eps)\cdot \log n\cdot \log (nw))$, which is bounded by $O(\log(1/\eps)+\log \log(nw)\cdot \log n\cdot \log (nw))$ since $O(\log(1/\eps))$ is the dominating term when $\log(1/\eps)\ge\log^3(nw)$.

\section{Preliminaries}\label{sec:prelim}
\subsection{Averaging samplers}
\begin{definition}
A function $g:\bits{n}\times\bits{d}\to\bits{m}$ is a \emph{$(\eps,\delta)$ (averaging) sampler} if for every function $f:\bits{m}\to{[0,1]}$, 
$$\pr[x\in\bits{n}]{\abs{\ex[s\in\bits{d}]{f(g(x,s))}-\ex[y\in\bits{m}]{f(y)}}\le\eps}\ge 1-\delta.$$
\end{definition}
\noindent
It's easy to show that samplers also work for $f$ with general range by scaling and shifting.
\begin{claim}\label{lemma:sampler-general}
Let $g:\bits{n}\times\bits{d}\to\bits{m}$ be a $(\eps,\delta)$-sampler, and let $\ell<r\in\bbR$. Then for every $f:\bits{m}\to{[\ell,r]}$, 
$$\pr[x\in\bits{n}]{\abs{\ex[s\in\bits{d}]{f(g(x,s))}-\ex[y\in\bits{m}]{f(y)}}\le\eps(r-\ell)}\ge 1-\delta.$$
\end{claim}
\begin{proof}
Let $f'$ be the function such that $f'(y)=(f(y)-\ell)/(r-\ell)$. Observe that the range of $f'$ is in $[0,1]$. By definition of sampler, 
$$\pr[x\in\bits{n}]{\abs{\ex[s\in\bits{d}]{f'(g(x,s))}-\ex[y\in\bits{m}]{f'(y)}}\le\eps}\ge 1-\delta.$$
By multiplying $(r-\ell)$ on both sides of the inequality inside the probability above we prove the claim.
\end{proof}
In our construction, we will use the following sampler which is explicitly computable with small space.
\begin{lemma}[\cite{RVW01,Gol11}]\label{lemma:samp-opt}
For every $\delta,\eps>0$ and integer $m$, there exists a $(\eps,\delta)$-sampler $f:\bits{n}\times\bits{d}\to\bits{m}$ s.t. $d=O(\log\log(1/\delta)+\log(1/\eps))$ and $n=m+O(\log(1/\delta))+O(\log(1/\eps))$. Moreover, for every $x,y$, $f(x,y)$ can be computed in space $O(m+\log(1/\delta)+\log(1/\eps))$.
\end{lemma}
\begin{remark}
The original sampler in \cite{RVW01} has a restriction on $\eps$. Such a restriction will cause a $2^{O(\log^*(nw/\eps))}$ factor in our construction, as in \cite{BCG18}. However, \cite{RVW01} pointed out that the restriction is inherited from the extractor in \cite{Zuc97}, which breaks down when the error is extremely small. As observed in \cite{Gol11}, this restriction can be removed by plugging in a more recent extractor construction in \cite{GUV09}. Note that there exists a space-efficient implementation of \cite{GUV09} in \cite{KNW08}, so the resulting sampler is also space-efficient. For completeness we include a proof in \Cref{sec:proof-of-samp}.
\end{remark}
\subsection{Matrix norms}
As in \cite{BCG18}, we will use the infinity norm in this paper.
\begin{definition}
For every matrix $A\in\bbR^{w\times w}$, $\ninf{A}=\max_i{\sum_j |A_{i,j}|}$.
\end{definition}
We record some well known properties of the infinity norm.
\begin{claim}
Let $A,B\in\bbR^{w\times w}$, $c\in \bbR$. Then
\begin{itemize}
\item
$\norm{cA}=|c|\norm{A}$
\item
$\norm{A}+\norm{B}\le\norm{A+B}$
\item
$\ninf{AB}\le \ninf{A}\ninf{B}$ 
\item
 $\max_{i,j} |A_{i,j}|\le\ninf{A}$
\item
If $A$ is stochastic, then $\ninf{A}=1$
\end{itemize}
\end{claim}
Note that for any $(n,w)$-ROBP represented by $w\times w$ matrices $M_{[n]}^{\bit}$, $\ninf{M_{i..j}}=1$ for every $0\le i\le j\le n$.

\section{Approximate Matrix Multiplication via Samplers}\label{sec:samp-our}

In this section we formally prove the sampler arguments which will be used in our construction. Our proof strategy resembles that of \cite{BCG18},  with the following two crucial differences. First, we will define two different notions of ``smallness" for our flattening idea. Second, in our construction we need the case where we use samplers to select matrices on both sides (\Cref{lemma:sym}). 

We will consider mappings $\cA:\bits{n}\to\bbR^{w\times w}$ which correspond to the implicit small norm matrices we discussed in the previous section. Borrowing notation from \Cref{def:MBS}, we use $\inp{\cA}$ to denote $\ex[x]{\cA(x)}$. First we define two different norms for the mapping $\cA$. The \emph{robust norm} is similar to the notion of ``smallness" in \cite{BCG18}, i.e. the average of norm of $\cA(x)$, while the \emph{norm} of $\cA$ is simply the norm of $\inp{\cA}$, i.e. the norm of average of $\cA(x)$. 
\begin{definition}
For every function $\cA:\bits{n}\to\bbR^{w\times w}$, we define the \emph{norm} of $\cA$ to be $\norm{\cA}=\ninf{\ex[x\in\bits{n}]{\cA(x)}}$, and the \emph{robust norm} of $\cA$ to be  $\norm[r]{\cA}=\ex[x\in\bits{n}]{\ninf{\cA(x)}}$. Besides, we define the weight of $\cA$ to be  $\wt(\cA)=\max_x \ninf{\cA(x)}$.
\end{definition}
\begin{claim}\label{claim:norms}
$\norm{\cA}\le\norm[r]{\cA}\le \wt(\cA)$.
\end{claim}
\begin{proof}
$\norm{\cA}\le\norm[r]{\cA}$ is by sub-additivity of $\norm{\cdot}$, and $\norm[r]{\cA}\le \wt(\cA)$ since $\norm[r]{\cA}$ is the average of values no larger than $\wt(\cA)$.
\end{proof}
\noindent
Next we show a simple lemma which will be used later. That is, a sampler for functions with range $[0,1]$ is also a sampler for matrix-valued functions, where the error is measured with infinity norm. 
\begin{lemma}\label{lemma:mat-samp-general}
For every function $\cA:\bits{m}\to\bbR^{w\times w}$ and every $(\eps,\delta)$-sampler $g:\bits{n}\times\bits{d}\to\bits{m}$,
$$\pr[x\in\bits{n}]{\ninf{\ex[s\in\bits{d}]{\cA(g(x,s))}-\inp{\cA}}\le 2w\wt(\cA)\eps}\ge 1- w^2\delta.$$
\end{lemma}
\begin{proof}
Let $\cE(y)=\cA(y)-\inp{\cA}$. For every $i,j\in[w]$, observe that 
$$\max_y \cE(y)_{i,j}-\min_y \cE(y)_{i,j}= \max_y \cA(y)_{i,j}-\min_y \cA(y)_{i,j}$$ 
By the property of sampler it follows that
$$\pr[x\in\bits{n}]{\abs{\ex[s]{\cE(g(x,s))_{i,j}}}\le 2\eps\wt(\cA)}\ge 1-\delta.$$ Using a union bound, 
$$\pr[x\in\bits{n}]{\forall i,j\in[w], \abs{\ex[s]{\cE(g(x,s))_{i,j}}}\le 2\eps\wt(\cA) }\ge 1- w^2\delta.$$
Thus by definition of the infinity norm, we can conclude that 
$$\pr[x\in\bits{n}]{\ninf{\ex[s\in\bits{d}]{\cE(g(x,s))}}\le 2w\wt(\cA)\eps }\ge 1-w^2\delta.$$
which by sub-additivity of $\ninf{\cdot}$ implies 
$$\pr[x\in\bits{n}]{\ninf{\ex[s]{\cA(g(x,s))}}\le \norm{\cA}+2w\wt(\cA)\eps}\ge 1- w^2\delta.$$
\end{proof}
\begin{corollary}\label{lemma:mat-samp}
For every function $\cA:\bits{m}\to\bbR^{w\times w}$ and every $(\eps,\delta)$-sampler $g:\bits{n}\times\bits{d}\to\bits{m}$,
$$\pr[x\in\bits{n}]{\ninf{\ex[s\in\bits{d}]{\cA(g(x,s))}}\le \norm{\cA}+2w\wt(\cA)\eps}\ge 1- w^2\delta.$$
\end{corollary}
\begin{proof}
By sub-additivity of $\ninf{\cdot}$, $\ninf{\ex[s\in\bits{d}]{\cA(g(x,s))}-\inp{A}}\le 2w\wt(\cA)\eps$ implies $\ninf{\ex[s\in\bits{d}]{\cA(g(x,s))}}\le \norm{\inp{\cA}}+2w\wt(\cA)\eps$. The claim now directly follows from \Cref{lemma:mat-samp-general}.
\end{proof}
\noindent
Now we introduce three different matrix multiplication rules. The first one is applying a sampler on both sides, and the second and third are applying sampler on only one side. 
\begin{lemma}[symmetric product]\label{lemma:sym}
Consider $\cA:\bits{n}\to\bbR^{w\times w}$ and $\cB:\bits{m}\to\bbR^{w\times w}$. Let $f:\bits{k}\times \bits{d_A}\to\bits{n}$ be a $(\delta,\eps_A)$ sampler, and $g:\bits{k}\times \bits{d_B}\to\bits{m}$ be a $(\delta,\eps_B)$ sampler. Then 
$$\ex[z]{\ninf{\ex[x,y]{\cA(f(z,x))\cB(g(z,y))}}}\le 2w^2\delta \wt(\cA)\wt(\cB)+\left(\norm{\cA}+2w\wt(\cA)\eps_A\right)\left(\norm{\cB}+2w\wt(\cB)\eps_B\right).$$
\end{lemma}
\begin{proof}
Let $$E_A=\left\{z: \ninf{\ex[x]{\cA(f(z,x)}}> \norm{\cA}+2w\wt(\cA)\eps_A\right\},$$and $$E_B=\left\{z: \ninf{\ex[y]{\cB(g(z,y)}}> \norm{\cB}+2w\wt(\cB)\eps_B\right\}.$$
Define $E=E_A\cup E_B$. By Lemma~\ref{lemma:mat-samp} and union bound, $\pr[z]{z\in E}<2w^2\delta$. Therefore
\begin{align*}
\ex[z]{\ninf{\ex[x,y]{\cA(f(z,x))\cB(g(z,y))}}}
&=\pr{z\in E}\ex[z\in E]{\ninf{\ex[x,y]{\cA(f(z,x))\cB(g(z,y))}}} \\
&+ \pr{z\not\in E}\ex[z\not\in E]{\ninf{\ex[x]{\cA(f(z,x))}\ex[y]{\cB(g(z,y))}}}\\
&\le 2w^2\delta \wt(\cA)\wt(\cB)+\ex[z\not\in E]{\ninf{\ex[x]{\cA(f(z,x))}}\ninf{\ex[y]{\cB(g(z,y))}}}\\
&\le 2w^2\delta \wt(\cA)\wt(\cB)+\left(\norm{\cA}+2w\wt(\cA)\eps_A\right)\left(\norm{\cB}+2w\wt(\cB)\eps_B\right).
\end{align*}
The second last inequality is by the fact that $\norm{\cdot}$ is non-negative and sub-multiplicative.
\end{proof}
\begin{lemma}[left product]\label{lemma:left}
Consider $\cA:\bits{k}\to\bbR^{w\times w}$ and $\cB:\bits{m}\to\bbR^{w\times w}$. Let $g:\bits{k}\times \bits{d_B}\to\bits{m}$ be a $(\delta,\eps_B)$ sampler. Then 
$$\ex[z]{\ninf{\ex[y]{\cA(z)\cB(g(z,y))}}}\le w^2\delta \wt(\cA)\wt(\cB)+\norm[r]{\cA}\left(\norm{\cB}+2w\wt(\cB)\eps_B\right).$$
\end{lemma}
\begin{proof}
Let $$E=\left\{z: \ninf{\ex[y]{\cB(g(z,y)}}> \norm{\cB}+2w\wt(\cB)\eps_B\right\}.$$
By Lemma~\ref{lemma:mat-samp}, $\pr[z]{z\in E}<w^2\delta$. Therefore 
\begin{align*}
\ex[z]{\ninf{\ex[y]{\cA(z)\cB(g(z,y))}}}
&=\pr{z\in E}\ex[z\in E]{\ninf{\ex[y]{\cA(z)\cB(g(z,y))}}} \\
&+ \pr{z\not\in E}\ex[z\not\in E]{\ninf{\ex[y]{\cA(z)\cB(g(z,y))}}}\\
&\le w^2\delta \wt(\cA)\wt(\cB)+\pr{z\not\in E}\cdot\ex[z\not\in E]{\ninf{\cA(z)}\ninf{\ex[y]{\cB(g(z,y))}}}\\
&\le w^2\delta \wt(\cA)\wt(\cB)+\pr{z\not\in E}\ex[z\not\in E]{\ninf{\cA(z)}}\cdot(\norm{\cB}+2w\wt(\cB)\eps_B)\\
&\le w^2\delta \wt(\cA)\wt(\cB)+\norm[r]{\cA}\left(\norm{\cB}+2w\wt(\cB)\eps_B\right).
\end{align*}
The third last inequality is by sub-multiplicativity of $\norm{\cdot}$, the second last inequality is by non-negativity of $\norm{\cdot}$, and the last inequality is by the fact that $$\pr{z\not\in E}\cdot\ex[z\not\in E]{\norm{\cA(z)}}=\ex[z]{\norm{\cA(z)}\cdot \mathbbm{1}(z\not\in E)}\le \norm[r]{\cA}.$$
\end{proof}
\begin{lemma}[right product]\label{lemma:right}
Consider $\cA:\bits{k}\to\bbR^{w\times w}$ and $\cB:\bits{m}\to\bbR^{w\times w}$. Let $f:\bits{k}\times \bits{d_A}\to\bits{n}$ be a $(\delta,\eps_A)$ sampler. Then 
$$\ex[z]{\ninf{\ex[x]{\cA(f(z,x))\cB(z)}}}\le w^2\delta \wt(\cA)\wt(\cB)+\left(\norm{\cA}+2w\wt(\cA)\eps_A\right)\norm[r]{\cB}.$$
\end{lemma}
\begin{proof}
Let $$E=\left\{z: \ninf{\ex[x]{\cA(f(z,x)}}> \norm{\cA}+2w\wt(\cA)\eps_A\right\}.$$
By Lemma~\ref{lemma:mat-samp}, $\pr[z]{z\in E}<w^2\delta$. Therefore 
\begin{align*}
\ex[z]{\ninf{\ex[x]{\cA(f(z,x))\cB(z)}}}
&=\pr{z\in E}\ex[z\in E]{\ninf{\ex[x]{\cA(f(z,x))\cB(z)}}} \\
&+ \pr{z\not\in E}\ex[z\not\in E]{\ninf{\ex[x]{\cA(f(z,x))\cB(z)}}}\\
&\le w^2\delta \wt(\cA)\wt(\cB)+\pr{z\not\in E}\cdot\ex[z\not\in E]{ \ninf{\ex[x]{\cA(f(z,x))}}\ninf{\cB(z)}}\\
&\le w^2\delta \wt(\cA)\wt(\cB)+(\norm{\cA}+2w\wt(\cA)\eps_A)\cdot \pr{z\not\in E}\ex[z\not\in E]{\ninf{\cB(z)}}\\
&\le w^2\delta \wt(\cA)\wt(\cB)+\left(\norm{\cA}+2w\wt(\cA)\eps_A\right)\norm[r]{\cB}.
\end{align*}
\end{proof}
\section{Main Construction}\label{sec:main-construction}
In this section we show our main construction and prove its correctness. We first introduce several definitions.
\begin{definition}
For every mapping $\cA:\bits{n}\to\bbR^{w\times w}$ and every matrix $A\in\bbR^{w\times w}$, we define $\cA-A$ to be the mapping s.t. $(\cA-A)(x)=\cA(x)-A$.
\end{definition}
\begin{definition}
Consider $A\in\bbR^{w\times w}$ and $\cA:\bits{n}\to\bbR^{w\times w}$. $\cA$ is a $\eps$-\emph{approximator} of $A$ if $\norm{\ex[x]{\cA(x)}-A}\le \eps$, i.e. $\norm{\cA-A}\le \eps$. $\cA$ is a $\eps$-\emph{robust approximator} of $A$ if $\ex[x]{\norm{\cA(x)-A}}\le \eps$, i.e. $\norm[r]{\cA-A}\le \eps$.
\end{definition}
\noindent
Now we define a robust PRPD. Note that a $(n,w,\eps)$-robust PRPD $(G,\rho)$ is also a $\wt(G,\rho)$-bounded $(n,w,\eps)$-PRPD.
\begin{definition}\label{def:robust-PRPD}
$(G,\rho):\bits{\sout}\times\bits{\si}\times{[\wt]}\to\bits{n}\times\bbR$ is a $(n,w,\eps)$-\emph{robust PRPD} if for every $(n,w)$-ROBP and its matrix representation $[M]_{[n]}^{\bit}$ the following holds. Let $\cA:\bits{\sout}\times\bits{\si}\to\bbR^{w\times w}$ denote the mapping 
$$\cA(x,y)=\ex[i\in {[\wt]}]{\rho(x,y,i)\cdot M_{0..n}^{G(x,y,i)}}.$$
\begin{itemize}
\item
Every $\rho(x,y,i)$ is either $\wt$ or $-\wt$. In other word, $\cA(x,y)$ is the summation of transition matrices with coefficient $\pm 1$.
\item
Let $\widehat{\cA}$ denote the mapping $\widehat{\cA}(x)=\ex[y]{\cA(x,y)}$. Then $\widehat{\cA}$ is a $\eps$-robust approximator for $M_{0..n}$. 
\end{itemize}
We say $\wt$ is the weight of $(G,\rho)$, denoted by $\wt(G,\rho)$. $\sout$ is the \emph{outer seed length} of $(G,\rho)$, denoted by $\sout(G,\rho)$. $\si$ is the \emph{inner seed length} of $(G,\rho)$, denoted by $\si(G,\rho)$. We write $s(G,\rho)=\sout(G,\rho)+\si(G,\rho)$ for short. We say $(G,\rho)$ is explicit if it can be computed in $O(s(G,\rho))$ space.

We say $\cA$ is the \emph{matrix form} of $(G,\rho)$ on $M_{0..n}$, and the definition of $\sout,\si,\wt$ on $(G,\rho)$ also apply to $\cA$. We say $\widehat{\cA}$ is the \emph{robust matrix form} of $(G,\rho)$ on $M_{0..n}$. 
\end{definition}
\begin{remark}\label{rmk:mon}
The above definition is similar to \Cref{def:MBS}, but each matrix $\cA(x,y)$ is realized with $\wt$ matrices instead of one matrix. These $\wt$ matrices will never be separated even after flattening. We do this in order to ensure that the matrix form always take bit-strings as input. This ensures that we can increase the outer and inner seed length of $\cA$ arbitrarily: we can construct the new mapping $\cA':\bits{\sout'}\times\bits{\si'}$ such that $\cA'(x,y)=\cA(x_p,y_p)$ where $x_p$ is the length-$\sout(\cA)$ prefix of $x$ and $y_p$ is the length-$\si(\cA)$ prefix of $y$. In other word, $\cA'$ computes the output only with prefix of necessary length of the input, and ignore the remaining bits. It is easy to verify that $\cA'$ is also the matrix form of a $(n,w,\eps)$-robust PRPD.
\end{remark}
\noindent
The following is some additional basic properties about robust PRPD and its flattened form.
\begin{claim}\label{claim:std-approx}
Let $(G,\rho):\bits{\sout}\times\bits{\si}\times{[\wt]}\to\bits{n}\times\bbR$ be a $(n,w,\eps)$-\emph{robust PRPD}. For every $(n,w)$-ROBP $M_1^0,M_1^1,\dots,M_n^0,M_n^1$ the following holds.
\begin{itemize}
\item 
Let $\widehat{\cA}$ be the robust matrix form of $(G,\rho)$ on $M_{0..n}$. Then $\wt(\widehat{\cA})\le \wt(G,\rho)$.
\item
Let $\cA$ denote the matrix form of $(G,\rho)$ on $M_{0..n}$. Let $\overline{\cA}:\bits{\sout+\si}\to\bbR^{w\times w}$ denote the mapping $\overline{\cA}(x\Vert y)=\cA(x,y)$. We say $\overline{\cA}$ is the \emph{flattened matrix form} of $(G,\rho)$ on $M_{0..n}$. Then $\overline{\cA}$ is an $\eps$-approximator for $M_{0..n}$, and $\wt(\overline{\cA})\le\wt(G,\rho)$.
\end{itemize}
\end{claim}
\begin{proof}
Recall that for every string $r\in\bits{n}$, $\norm{M_{0..n}^r}=1$. By sub-additivity of $\norm{\cdot}$ we have $\norm{\cA(x,y)}\le\wt(G,\rho)$ for every $x,y$, which implies $\wt(\overline{\cA})\le\wt(G,\rho)$. By sub-additivity and scalibility of $\norm{\cdot}$, we have $\wt(\cA')\le \wt(\cA)$. To show that $\overline{\cA}$ is a $\eps$-approimxator of $M_{0..n}$, observe that $\cA'$ is also an $\eps$-approximator of $M_{0..n}$ by \Cref{claim:norms}, and note that $\inp{\cA}=\inp{\cA'}$.
\end{proof}
\noindent
Now we prove our main lemma. The following lemma allows us to construct robust PRPDs for $(2m,w)$ ROBPs from robust PRPDs for $(m,w)$ ROBPs, without increasing the seed length too much. We will recursively apply this lemma for $\log n$ levels to get a $(n,w,\eps)$-robust PRPD. The basic idea is as described in \Cref{lemma:main-informal}.
\begin{lemma}\label{lemma:main-approx}
Suppose there exists $\sout,\si$ such that the following conditions hold.
\begin{itemize}
\item
For every $0\le i\le k$, there exists a $(m,w,\gamma^{i+1})$-robust PRPD $(G_i,\rho_i)$ s.t. $\wt(G_i,\rho_i)\le \binom{m-1}{i}$ and $\sout(G,\rho)\le \sout$. Moreover, for every $0\le i\le \ceil{k/2}$, $s(G_i,\rho_i) \le \sout$.
\item
For every $i\le \ceil{k/2}$, there exists a $(\eps_i,\delta)$-sampler $g_i:\bits{\sout}\times\bits{d_i}\to\bits{s(G_i,\rho_i)}$, where $\eps_i\le \gamma^{i+1}/(w\cdot\binom{m-1}{i})$ and $\delta\le \gamma^{k+1}/(w^2\cdot \binom{2m-1}{i})$.
\item
For every $i\ge j\ge 0$ s.t. $i+j\le k$, if $j\le i\le \ceil{k/2}$, then $d_i+d_j\le \si$. If $i>\ceil{k/2}$, then $\si(G_i,\rho_i)+d_j\le \si $.
\end{itemize}
Then there exists a $(2m,w,(11\gamma)^{k+1})$-robust PRPD $(G,\rho)$ s.t. $\sout(G,\rho)=\sout$, $\si(G,\rho)=\si$ and $\wt(G,\rho)\le \binom{2m-1}{k}$.
\end{lemma}
\begin{proof}
Fix any $(2m,w)$-ROBP with matrix representation $M_{[2m]}^{\bit}$. Let $A=M_{0..m}$ and $B=M_{m..2m}$.
 For every $0\le i\le k$, let $\cA_i,\widehat{\cA}_i,\overline{\cA}_i$ denote the matrix form, robust matrix form and flattened matrix form of $(G,\rho)$ on $M_{0..m}$ respectively. Let $\cB_i,\widehat{\cB}_i,\overline{\cB}_i$ denote the matrix form, robust matrix form and flattened matrix form of $(G,\rho)$ on $M_{m..2m}$ respectively. By definition, $\widehat{\cA}_i$ and $\widehat{\cB}_i$ are $\gamma^{i+1}$-robust approximator for $A$ and $B$ respectively. By \Cref{claim:std-approx}, $\overline{\cA}_i$ and $\overline{\cB}_i$ are $\gamma^{i+1}$-approximator for $A$ and $B$ respectively. Moreover, we will increase the outer seed length of $\cA_i$ and $\cB_i$ to match the length of the given input when necessary. (See \Cref{rmk:mon})
 
Now for every $x,y$ we define a mapping $\cC_k:\bits{\sout}\times\bits{\si}\to\bbR^{w\times w}$ as follows. Note that $\cC_k$ corresponds to the matrix form of $(G,\rho)$ on $M_{0..2m}$.
\begin{enumerate}[(1)]
\item
For every $0\le i\le \ceil{k/2}$, let $a_i$ be the prefix of $y$ of length $d_i$ and $b_i$ be the suffix of $y$ of length $d_i$. Define $A_{x,y,i}=\overline{\cA_i}(g_i(x,a_i))$ and $B_{x,y,i}=\overline{\cB_i}(g_i(x,b_i))$.
\item
For every $\ceil{k/2}<i\le k$, let $a_i$ be the prefix of $y$ of length $\si(\cA_i)$ and $b_i$ be the suffix of $y$ of length $\si(\cB_i)$. Define $A_{x,y,i}=\cA_i(x,a_i)$ and $B_{x,y,i}=\cB_i(x,b_i)$.
\item
Define $\cC_k(x,y)=\sum_{i+j=k} A_{x,y,i}B_{x,y,j}-\sum_{i+j=k-1}A_{x,y,i}B_{x,y,j}$.
\end{enumerate}
Note that for every $i+j\le k$, prefix $a_i$ and suffix $b_j$ of $y$ never overlap.

By expanding every $A_{x,y,i}B_{x,y,j}$ term with distributive law, we can see that each small term in $A_{x,y,i}B_{x,y,j}$ has coefficient $\pm 1$, which satisfies the first condition of robust PRPD. Moreover, the total number of terms after expanding is 
$$\wt(\cC_k)\le \sum_{i+j=k}\binom{m-1}{i}\cdot \binom{m-1}{j} + \sum_{i+j=k-1}\binom{m-1}{i}\cdot \binom{m-1}{j}= \binom{2m-1}{k}.$$
It remains to show that $\cC_k$ satisfies the second condition of robust PRPD, i.e. $\ex[y]{\cC_k(x,y)}$ is a good approximation of $M_{0..2m}=AB$ on average over $x$. Observe that 
\begin{align*}
\ex[x]{\norm{\ex[y]{\cC_k(x,y)}-AB}}
&=\ex[x]{\norm{\ex[y]{\cC_k(x,y)-AB}}}\\
&\le \sum_{i+j=k} \ex[x]{\norm{\ex[y]{(A_{x,y,i}-A)(B_{x,y,j}-B)}}} \\
&+ \sum_{i+j=k-1} \ex[x]{\norm{\ex[y]{(A_{x,y,i}-A)(B_{x,y,j}-B)}}}\\
&+ \ex[x]{\norm{\ex[y]{(A_{x,y,k}-A)B}}} + \ex[x]{\norm{\ex[y]{A(B_{x,y,k}-B)}}},
\end{align*}
by decomposing $\cC_k(x,y)-AB$ with the equation in the proof of \Cref{lemma:main} and applying sub-additivity of $\norm{\cdot}$. 

First we consider the last two terms. Since $\norm{B}=1$, by sub-multiplicativity we have 
$$\ex[x]{\norm{\ex[y]{(A_{x,y,k}-A)B}}}\le \ex[x]{\norm{\ex[y]{A_{x,y,k}-A}}}.$$
Now consider two cases. If $k\ge 2$, then 
$$\ex[x]{\norm{\ex[y]{A_{x,y,k}-A}}}=\ex[x]{\norm{\widehat{\cA}_k(x)-A}}\le \gamma^{k+1}$$ by definition. If $k<2$, then 
$$\ex[x]{\norm{\ex[y]{A_{x,y,k}-A}\cdot B}}=\ex[x]{\norm{\ex[a_k]{\overline{\cA_k}(g_k(x,a_k))-A}}\cdot B}.$$
Apply \Cref{lemma:right} on $\overline{\cA_k}-A$ and the dummy mapping $\cB$ s.t. $\cB(x)=B$ for every $x$, we can derive that the above formula is bounded by $w^2\delta\binom{m-1}{k}+3\gamma^{i+1}$. For the term $\ex[x]{\norm{\ex[y]{A(B_{x,y,k}-B)}}}$ we can get the same bound with a similar proof.

Now consider the terms in the form $\ex[x]{\norm{\ex[y]{(A_{x,y,i}-A)(B_{x,y,j}-B)}}}$. \\
First consider the case $i,j\le\ceil{k/2}$. Then
\begin{align*}
&\ex[x]{\norm{\ex[y]{(A_{x,y,i}-A)(B_{x,y,j}-B)}}}\\
&=\ex[x]{\norm{\ex[a_i]{\overline{\cA_i}(g_i(x,a_i))-A}\ex[b_j]{\overline{\cB_k}(g_j(x,b_j))-B}}}\textrm{ (since }a_i,b_j\textrm{ don't overlap)}\\
&\le 2w^2\delta \cdot \binom{m-1}{i}\cdot\binom{m-1}{j}+9\gamma^{i+j+2}.\textrm{ (by \Cref{lemma:sym})}
\end{align*}
Next consider the case $i>\ceil{k/2},j\le \ceil{k/2}$. Then 
\begin{align*}
&\ex[x]{\norm{\ex[y]{(A_{x,y,i}-A)(B_{x,y,j}-B)}}}\\
&=\ex[x]{\norm{\widehat{\cA}_i(x)\cdot\ex[b_j]{\overline{\cB_k}(g_j(x,b_j))-B}}}\textrm{ (since }a_i,b_j\textrm{ don't overlap)}\\
&\le w^2\delta \cdot \binom{m-1}{i}\cdot\binom{m-1}{j}+3\gamma^{i+j+2}.\textrm{ (by \Cref{lemma:right})}
\end{align*}
Similarly for the case that $i\le \ceil{k/2},j>\ceil{k/2}$ we can show that 
$$\ex[x]{\norm{\ex[y]{(A_{x,y,i}-A)(B_{x,y,j}-B)}}}\le w^2\delta \cdot \binom{m-1}{i}\cdot\binom{m-1}{j}+3\gamma^{i+j+2}$$
by \Cref{lemma:left}. Finally, note that the case $i,j>\ceil{k/2}$ does not exist because $i+j\le k$. 

Taking the summation of all the cases, we get 
\begin{align*}
&\ex[x]{\norm{\ex[y]{\cC_k(x,y)}-AB}}\\
&\le 2w^2\delta\cdot\left(\sum_{i+j=k}\binom{m-1}{i}\binom{m-1}{j}+\sum_{i+j=k-1}\binom{m-1}{i}\binom{m-1}{j}+\binom{m-1}{k}\right)\\
&+ (k+1)\cdot 9\gamma^{k+2}+ k\cdot 9\gamma^{k+1} + 2\cdot 3\gamma^{k+1} \\
&\le 4w^2\delta \cdot \binom{2m-1}{k} + (9k+9)\gamma^{k+2} + (9k+6)\gamma^{k+1}\\
&\le (10k+11)\gamma^{k+1}\\
&\le (11\gamma)^{k+1}.
\end{align*}
Moreover, note that $AB=M_{0..2m}$, and the construction of $\cC_k$ does not depend on the matrices $M_{[2m]}^{\bit}$. (See \Cref{sec:matrices} for how the arithmetic operations in $\cC_k(x,y)$ are translated back to operations on pseudo-distributions.) Therefore there exists a $(2m,w,(11\gamma)^{k+1})$-robust PRPD $(G,\rho)$. 
\end{proof}
\noindent
Finally we analyze the seed length of the recursive construction, and present the main theorem.
\begin{theorem}
There exists an explicit $(n,w,\eps)$-robust PRPD $(G,\rho)$ such that 
\begin{itemize}
\item 
$\sout(G,\rho)= O\left(\log(1/\eps)+ \log n \log (nw)\log\left(\frac{\log(1/\eps)}{\log n}\right)\right)$
\item
$\si(G,\rho)= O\left(\log(1/\eps)+ \log n \log (nw)\log\left(\frac{\log(1/\eps)}{\log n}\right)\right)$
\item
$\wt(G,\rho)=\poly(1/\eps)$
\end{itemize}
Moreover, for every $B$ the approximator $\cG$ has the same corresponding pseudodistribution.
\end{theorem}
\begin{proof}
Let $c$ be the constant such that for every $\eps,\delta>0$ there exists a $(\eps,\delta)$-sampler $g:\bits{n}\times\bits{d}\to\bits{m}$ such that $n=m+c\log(1/\eps)+c\log(1/\delta)$ and $d=c\log (1/\eps)+c\log\log(1/\delta)$, as guaranteed in \Cref{lemma:samp-opt}. WLOG assume that $n$ is a power of $2$. Define $\gamma=1/n^{4}$. For every $0\le h \le \log n$, every $k\ge 0$, we will inductively prove that there exists a $(2^h,w,(11^h\gamma)^{k+1})$-robust PRPD $(G_{h,k},\rho_{h,k})$ with the following parameters.
\begin{itemize}
\item
If $k\le 1$, $\sout(G_{h,k},\rho_{h,k})\le h \cdot (3ck\log(n/\gamma) + 7c\log(w/\gamma))$
\item 
If $k>1$, $\sout(G_{h,k},\rho_{h,k})\le 4ck\log (n/\gamma)+(\ceil{\log k}+1)\cdot h \cdot (10c\log(w/\gamma))$
\item
If $k\le 1$, $\si(G_{h,k},\rho_{h,k})\le ck\log(n/\gamma) + 4c\log(w/\gamma)$
\item
If $k>1$, $\si(G_{h,k},\rho_{h,k})\le ck\log (n/\gamma)+ h \cdot (4c\log(kw/\gamma))$
\item
$\wt(G_{h,k},\rho_{h,k})\le \max(1,\binom{2^h-1}{k})$
\end{itemize}
We will write $\sout[,h,k]=\sout(G_{h,k},\rho_{h,k})$ and $\si[,h,k]=\sout(G_{h,k},\rho_{h,k})$ for short. First consider the terminal case $2k\ge 2^h$ or $h=0$. In this case we simply take $\sout[,h,k]=0$, $\si[,h,k]=2^h \le 2k$ and $\wt(G_{h,k},\rho_{h,k})=1$ s.t. $G_{h,k}(x,y,i)=y$ and $\rho_{h,k}(x,y,i)=1$. For the other cases, we show that we can get the intended parameters by constructing $(G_{h,k},\rho_{h,k})$ with the recursion in \Cref{lemma:main-approx}. Note that based on the induction hypothesis we can assume $\cG_{a,h-1,k}$ and $\cG_{a+2^{h-1},h-1,k}$ have exactly the same parameters, so we consider the parameter of $\cG_{a,h-1,k}$ only. We have seen that the bound for $\wt(G_{h,k},\rho_{h,k})$ is correct. First we show that the bound for $\si[,h,k]$ is correct. Recall that in the recursion we take parameters $d_i=c\log(1/\eps_i)+c\log\log(1/\delta)\le ci\log(n/\gamma)+2c\log(knw/\gamma)$, based on the fact that $\binom{2^h-1}{i}\le n^i$. Now consider the restriction on $\si(\cG_k)$ in our recursion. For $i+j\le k$ and $j\le i\le\ceil{k/2}$, we need 
$$d_i+d_j\le ck\log(n/\gamma)+4c\log(knw/\gamma)\le\si[,h,k]$$ which is true. For $i+j\le k$ and $i>\ceil{k/2}$, we need 
\begin{align*}
\si[,h-1,i]+d_j
&\le ci\log(1/\gamma)+(h-1)\cdot 4c\log(inw/\gamma)+(cj\log(1/\gamma)+2c\log(knw/\gamma))\\
&\le ck\log(1/\gamma) +h\cdot 4c\log(knw/\gamma)\\
&\le \si[,h,k]
\end{align*}
which is also true. Moreover, observe that when $k\le 1$ it is always the case that $i,j\le\ceil{k/2}$. Therefore the third condition is also true. Finally we show that the bound for $\sout[,h,k]$ is also correct. First observe that the restriction $\sout[,h-1,i]\le \sout[,h,k]$ is trivially true. Then the only condition left is that for every $i\le\ceil{k/2}$,
$$\sout[,h-1,i]+\si[,h-1,i]+c\log(1/\delta)+c\log(1/\eps_i)\le \sout[,h,k].$$ Since $\sout[,h-1,i]\le \sout[,h-1,\ceil{k/2}]$ and $\si[,h-1,i]\le \si[,h-1,\ceil{k/2}]$ for every $i$, it suffices to show that 
$$\sout[,h-1,\ceil{k/2}]+\si[,h-1,\ceil{k/2}]+c\log(1/\delta)+c\log(1/\eps_{\ceil{k/2}})\le \sout[,h,k].$$
First we consider $k\le 1$, which is the case that $\ceil{k/2}=k$. Then 
\begin{align*}
&\sout[,h-1,\ceil{k/2}]+\si[,h-1,\ceil{k/2}]+c\log(1/\delta)+c\log(1/\eps_{\ceil{k/2}})\\
&\le \sout(\cG_{a,h-1,k}) + 3ck\log(n/\gamma) + 7c\log(n/\gamma)\\
&\le h \cdot (3ck\log(n/\gamma) + 7c\log(n/\gamma)) \\
&\le  \sout[,h,k].
\end{align*}
Finally we consider the case $k>1$. Observe that 
\begin{align*}
&\sout[,h-1,\ceil{k/2}]+\si[,h-1,\ceil{k/2}]+c\log(1/\delta)+c\log(1/\eps_{\ceil{k/2}})\\
&\le \sout[,h-1,\ceil{k/2}]+\si[,h-1,\ceil{k/2}] + \frac{3k+1}{2}\cdot c\log(n/\gamma)+3c\log(w/\gamma)\\
&\le \sout[,h-1,\ceil{k/2}]+ (2k+1)\cdot c\log(n/\gamma) + (h-1)\cdot 4c\log(w/\gamma) + 7c\log(w/\gamma) \\
&\le  4c\cdot\frac{k+1}{2}\cdot\log(n/\gamma) + \left(\ceil{\log \ceil{\frac{k}{2}}}+1\right)\cdot (h-1) \cdot (10c\log(nw/\gamma)) \\
&+ (2k+1)\cdot c\log(n/\gamma) + (h-1)\cdot 4c\log(w/\gamma) + 7c\log(w/\gamma) \\
&\le 4ck\log(n/\gamma) + \left(\ceil{\log \ceil{\frac{k}{2}}}+1\right)\cdot (h-1) \cdot (10c\log(nw/\gamma)) + h\cdot 10c\log(nw/\gamma)  \\
&\le \sout[,h,k].
\end{align*}
In the last inequality we use the fact that $\ceil{\log k}=\ceil{\log(\ceil{k/2})}+1$ for every $k>1$.

Finally, note that $(11^{\log n}\gamma)=n^{\log_2 11}\cdot n^{-4}\le n^{-0.5}$. By taking $h=\log n$ and  $k=\frac{\log(1/\eps)}{\log (1/n^{0.5})}$, we get a $(n,w,\eps)$-robust PRPD.
\end{proof}

\begin{remark}
To get the seed length we claimed in \Cref{thm:main}, observe that the $\log(1/\eps)$ term is dominating when $\log(1/\eps)\ge \log^3(nw)$. Therefore we can simply replace the $\log\log(1/\eps)$ factor on the $O(\log n\log (nw))$ term with $\log\log(nw)$.
\end{remark}

\section{Discussion and Open Questions}\label{sec:open-questions}
We discuss some natural questions that arise from our work.
\begin{itemize}
\item 
In our construction, we applied the sampler argument in \cite{BCG18} without constructing small-norm matrices explicitly. This is probably hinting that negative weight is not essentially required for the sampler argument. Is it possible to apply the sampler argument to construct a PRG (instead of PRPD) with improved dependency on error?
\item
Is there an explicit PRPD which matches the seed length of the hitting set generator in \cite{HZ18}, i.e. $O(\log(w/\eps))$ when $n=\poly\log(w)$? A possible direction is to adapt our construction to a $t$-ary recursion tree where $t=\log^{1-\Omega(1)} (n)$ instead of a binary tree, as in \cite{NZ96,Arm98}. However, a direct adaption requires us to apply samplers on $(t-1)$-children in each recursion, and for every sampler we need to pay some randomness for ``inner seed" which cannot be recycled. In our construction we see that the inner seed of a sampler contains a $\log w$ term. Therefore in each recursion we need to pay at least $(t-1)\log w$ which is too expensive. Is it possible to make the sampler argument work with a shorter inner seed?
\item
Is it possible to improve the seed length to $\tilde{O}(\log^2 n + \log (w/\eps))$, even in some restricted settings? We note that there are two things which cause the $\Omega(\log n\cdot\log w)$ term in our construction. The first one is the inner seed of sampler, which is related to the question above. The second one is the restriction on the outer seed length, which is analogous to ``entropy loss" if we view the samplers as extractors. Note that \cite{RR99} shows how to ``recycle entropy" in the INW generator in some restricted settings, but it is not clear how to apply the extractor-type analysis of INW generator in our construction.  
\end{itemize}


\bibliography{ref}
\bibliographystyle{alpha}
\appendix
\section{Using PRPDs in the Saks-Zhou Scheme}\label{sec:Saks-Zhou}
In this section, we will briefly introduce Saks and Zhou's proof for $\BPL\subseteq \L^{3/2}$~\cite{SZ99} and Armoni's trick for replacing Nisan's PRG with any PRG in this proof~\cite{Arm98}. Then we will see why a $\poly(nw/\eps)$-bounded PRPD suffices for this scheme. Since our purpose here is to go over the possible difference between using PRGs and PRPDs in this scheme, we will only include a sketch of Saks and Zhou's proof. We recommend interested readers to check \cite{SZ99,Arm98} for formal proofs and also \cite{HU17} for a beautiful summary.
\subsection{Saks and Zhou's Scheme}
It is well-known that derandomizing $\BPL$ can be reduced to approximating $M^{n}$ where $M$ is any $n\times n$ stochastic matrix. The first step of Saks and Zhou is to turn $M^n$ into the following recursive computation:
\begin{fact}
Let $n_1,n_2$ be integers such that $n_1^{n_2}=n$. Define $M_0=M$, and $M_i=M_{i-1}^{n_1}$ for every positive integer $i$. Then $M_{n_2}=M^n$.
\end{fact}
To approximate $M_{n_2}$, it suffices to compute $M_i=M_{i-1}^{n_1}$ with small enough error in each step. However, if we need $s$ bits of space to approximate the $n_1$-th power of a stochastic matrix, we will need $O(sn_2)$ bits of space in total. This doesn't really save any space (over approximating $M^n$ directly) if we approximate $M_{i-1}^{n_1}$ with PRGs such as Nisan's generators. The first idea of Saks and Zhou is to utilize the ``high probability" property of Nisan's generator: 
\begin{lemma}[\cite{Nis92}]
For every $n,w,\eps$ there exists an algorithm $\widehat{\mathrm{Pow}}_{n}$ which takes a $w\times w$ (sub)stochastic matrix $M$ and a string $y\in\bits{O(\log n \log(nw/\eps))}$ as input, and outputs a $w\times w$ matrix such that 
$$\pr[y]{\nmax{\widehat{\mathrm{Pow}}_n(M,y)-M^n}<\eps}\ge 1-\eps$$
in space $O(log(nw/\eps))$.
\end{lemma}
\noindent
In other word, derandomization with Nisan's generator has the following structure. First it fixes an ``offline randomness" $y\in\bits{r}$ and considers it as a part of input. Then it takes $s$ bits of additional ``processing space" to compute an approximation of $M^n$. Then the output will be a good approximation with high probability over $y$. (This is called an ``offline randomized algorithm" in \cite{SZ99}.) Furthermore $s\ll r$. With these properties, the main idea of Saks and Zhou is to reuse the same offline randomness for each level of recursion. If computing $M^{n_1}$ takes $r$ bits of offline randomness and $s$ bits of processing space, then computing $M^n$ will take $r$ bits of offline randomness and $O(sn_2)$ bits of processing space. The space complexity would be $O(r+sn_2)$ which is better than approximating $M^n$ directly since the offline randomness part was the original bottleneck. 

However, there's a problem in this construction: if we compute $\widehat{M_1}=\widehat{\mathrm{Pow}}_{n_1}(M,y)$, and try to use $\widehat{\mathrm{Pow}}_{n_1}(\widehat{M_1},y)$ to approximate $M_2=M_1^{n_1}$, it might be possible that $\widehat{\mathrm{Pow}}_{n_1}(\widehat{M_1},y)$ is always a bad approximation because $\widehat{M_1}$ depends on $y$. To resolve this issue, the second idea of Saks and Zhou is to break the dependency with a randomized rounding operation. We will borrow the name ``snap" from \cite{HU17} for this operation.
\begin{definition}
Given value $x\in\bbR$, string $y\in \bits{d}$, define $$\mathrm{Snap}_{d}(x,y)=\max(\floor{x\cdot 2^d-2^{-d}y}\cdot 2^{-d},0).$$ For a $w\times w$ matrix $M$, define $\mathrm{Snap}_{d}(M,y)$ to be the matrix $M'$ such that $M'_{i,j}=\mathrm{Snap}_{d}(M_{i,j},y)$ for every $i,j\in [w]$.
\end{definition}
\noindent
In other word, in a snap operation, we randomly perturb the matrix with a offset in $[0,2^{-2d}]$, then round the entries down to $d$ bits of precision. It's not hard to prove the following lemma:
\begin{lemma}[\cite{SZ99}]
For any matrix $M,M'$ such that $\nmax{M-M'}\le \eps$, $$\pr[y]{\mathrm{Snap}_{d}(M,y)\neq \mathrm{Snap}_{d}(M',y)}\le w^2 (2^d\eps+2^{-d}).$$
\end{lemma}
\begin{proof}
The snap operation is equivalent to randomly choose a grid of length $2^{-d}$ and round each value to the closest grid point the left. Therefore two values $a,b$ are rounded to different points only if there is a grid point between them, which happens with probability at most $2^d|a-b|+2^{-d}$. By union bound and the fact that $\nmax{M-M'}\le \eps$ the lemma follows.
\end{proof}

With the lemma above, we can see that by taking $\widehat{M_1}=\mathrm{Snap}_d(\widehat{\mathrm{Pow}}_{n_1}(M,y),z)$ instead, $\widehat{M_1}$ will be equivalent to $\mathrm{Snap}_d(M^{n_1},z)$ with high probability, which is independent of $y$. Therefore we can use $y$ as the offline randomness to compute the $n_1$-th power of $\widehat{M_1}$. Moreover, if the rounding precision is high enough, the snapped matrix is still a good approximation. Finally we get Saks-Zhou theorem:
\begin{lemma}[\cite{SZ99}]\label{lemma:Saks-Zhou}
Let $n_1,n_2$ be integers such that $n_1^{n_2}=n$. Suppose there exists an offline randomized algorithm $\widehat{\mathrm{Pow}}_{n_1}$ which takes $r$ bits of randomness and $s$ bits of processing space such that for every substochastic matrix $M$, 
$$\pr[x]{\nmax{\widehat{\mathrm{Pow}}_{n_1}(M,y)-M^{n_1}}\le \eps}\ge 1-\eps.$$
Now consider uniform random bits $y\in\bits{r}$ and $z_1,z_2,\dots,z_{n_2}\in\bits{d}$. 
Let $\widehat{M_0}=M$, and $\widehat{M_i}=\mathrm{Snap}_d(\widehat{\mathrm{Pow}}_{n_1}(\widehat{M_{i-1}},y),z_i)$ for every $i\in[n_2]$. Then with probability at least $1-O(w^2 n_2(2^d\eps+2^{-d}))$ over $y,z_1,\dots,z_{n_2}$,
$$\norm{\widehat{M_{n_2}}-M^n}\le nw2^{-d+1}.$$
Moreover, the space complexity of computing $\widehat{M_{n_2}}$ is $O(r+n_2(s+d))$.
\end{lemma}
\begin{proof}[Proof] (sketch)
Define $\overline{M_0}=M$, $\overline{M_i}=\mathrm{Snap}((\overline{M_{i-1}})^{n_1},z_i)$. By union bound, the following events happen simultaneously with probability $1-O(w^2 n_2(2^d\eps+2^{-d}))$:
\begin{enumerate}
\item 
For every $i\in[n_2]$, $\nmax{\widehat{\mathrm{Pow}}_{n_1}(\overline{M_{i-1}},y)-(\overline{M_{i-1}})^{n_1}}\le \eps$.
\item
For every $i\in[n_2]$, conditioned on $\widehat{M_{i-1}}=\overline{M_{i-1}}$ and $\nmax{\widehat{\mathrm{Pow}}_{n_1}(\overline{M_{i-1}},y)-\overline{M_{i-1}}^{n_1}}\le \eps$, $\widehat{M_{i}}=\overline{M_{i}}$.
\end{enumerate}
When the above events occur, we have $\widehat{M_{n_2}}=\overline{M_{n_2}}$. Moreover, note that for every $i\in[n_2]$ $$\nmax{\overline{M_{i}}-(\overline{M_{i-1}})^{n_1}}\le 2^{-d+1}.$$
To see why this is true, observe that in a snap operation we change the given value by at most $2^{-2d}$ from perturbation and $2^{-d}$ from rounding. \footnote{Note that capping the lowest possible value to be $0$ can only reduce the error, because the snapped value was non-negative.}
This implies 
$$\norm{\overline{M_{i}}-(\overline{M_{i-1}})^{n_1}}\le 2^{-d+1},$$
where $\norm{\cdot}$ denotes the matrix infinity norm.
By Lemma 5.4 in \cite{SZ99},
$$\norm{\overline{M_{n_2}}-M^n}\le nw2^{-d+1}.$$
For the space complexity, observe that we can compute $\widehat{M}_{n_2}$ with $n_2$ levels of recursive calls, and each recursive call takes $O(s+d)$ bits. Moreover, we need $r$ bits to store the offline randomness. Therefore the space complexity is $O(n_2(s+d)+r)$
\end{proof}
If we take $n_2=\sqrt{\log n}$, $n_1=2^{\sqrt{\log n}}$, $d=O(\log(n))$ and $\eps=2^{-2d}$ and plugging in Nisan's generator, the above lemma shows that $\BPL\subseteq \L^{3/2}$.

\subsection{Armoni's Trick}
We saw that in Saks and Zhou's proof, we need a ``offline randomized algorithm" for substochastic matrix exponentiation such that when given $r$ bits of randomness as additional input, the algorithm only requires additional $s\ll r$ bits of space to compute a good approximation with high probability. This is in fact the only place where we need PRGs in Saks and Zhou's proof. However, not every PRG has such property, so it might be hard to tell whether an improvement over Nisan's PRG will actually give a better derandomization for $\BPL$. Fortunately, Armoni \cite{Arm98} observed that one can turn any PRG into a derandomization algorithm with the required property by simply composing the PRG with an averaging sampler.

Before we go through Armoni's claim, first we generalize \Cref{lemma:Saks-Zhou} for a larger class of algorithms $\widehat{\mathrm{Pow}}$. 
\begin{definition}
We say an offline randomized algorithm requires $s$ bits of \emph{sensitive processing space} and $t$ bits of \emph{reusable processing space} if 
\begin{itemize}
\item 
During the execution of this algorithm, only $t$ bits of processing space is required.
\item
Before each time a bit is read from the real input (not including the offline randomness), only $s$ bits of processing space is being used at the time. 
\end{itemize}
\end{definition}
In the above definition, think of each input bit as generated from a recursive call. Thus the ``reusable processing space" can be interpreted as ``recursion-friendly processing space", which can be erased before every recursive call. With this new definition we can generalize \Cref{lemma:Saks-Zhou} as follows:
\begin{lemma}[\cite{SZ99}, generalized]\label{lemma:SZ-general}
Let $n_1,n_2$ be integers such that $n_1^{n_2}=n$. Suppose there exists an offline randomized algorithm $\widehat{\mathrm{Pow}}_{n_1}$ which takes $r$ bits of randomness, $s$ bits of sensitive processing space and $t$ bits of reusable processing space, such that for every substochastic matrix $M$, 
$$\pr[x]{\nmax{\widehat{\mathrm{Pow}}_{n_1}(M,y)-M^{n_1}}\le \eps}\ge 1-\eps.$$
Now consider uniform random bits $y\in\bits{r}$ and $z_1,z_2,\dots,z_{n_2}\in\bits{d}$. 
Let $\widehat{M_0}=M$, and $\widehat{M_i}=\mathrm{Snap}_d(\widehat{\mathrm{Pow}}_{n_1}(\widehat{M_{i-1}},y),z_i)$ for every $i\in[n_2]$. Then with probability at least $1-O(w^2 n_2(2^d\eps+2^{-d}))$ over $y,z_1,\dots,z_{n_2}$,
$$\norm{\widehat{M_{n_2}}-M^n}\le nw2^{-d+1}.$$
Moreover, the space complexity of computing $\widehat{M_{n_2}}$ is $O(r+t+n_2(s+d))$.
\end{lemma}
We omit the proof because it's Exactly the same as \Cref{lemma:Saks-Zhou}.

For technicality, we also need to define a ROBP with larger ``step size".
\begin{definition}
A $(n,w,d)$-ROBP is a ROBP of $n$ layers, $w$ nodes in each layer, and $2^d$ branches from each node. 
\end{definition}
That is, a $(n,w,d)$-ROBP is a ROBP which can read $d$ bits at once. Note that derandomizing $(n,w,d)$-ROBP corresponds to derandomizing the exponentiation of a stochastic matrix which has $d$ bits of precision in each entry.

Now we are ready to introduce Armoni's Lemma.
\begin{lemma}[\cite{Arm98}]\label{lemma:Armoni}
Suppose there exists an explicit PRG for $(n,w+1,\log(3nw/\eps))$-ROBP with error $\eps/3$ which has seed length $s$. Then there exists an offline randomized algorithm which approximates the $n$-th power of any substochastic matrix within error $\eps$ with probability at least $1-\eps$. Moreover, such algorithm requires $s+O(\log(w/\eps))$ bits of randomness, $O(s+O(\log(w/\eps)))$ bits of reusable processing space and $O(\log(nw/\eps))$ bits of sensitive processing space. 
\end{lemma}
\begin{proof}
Given an input $M$, first we round each entry down to $d=\log(3nw/\eps)$ bits of precision. Then we will get a substochastic matrix $M'$ such that each entry of $M'$ is a multiple of $2^{-d}$, and $\nmax{M-M'}\le \eps/3nw$. Then we have 
$$\nmax{M^n-(M')^n}\le \norm{M^n-(M')^n}\le n \norm{M-M'}\le nw\nmax{M-M'}\le \frac{\eps}{3}.$$

Then we construct a $(n,w+1,d)$-ROBP $B$ as follows. For each $t\in[n]$, we connect $k$ edges from node $(t-1,i)$ to node $(t,j)$ if $M'_{i,j}=k\cdot 2^{-d}$. Then for each node $(t-1,i)$ which doesn't have $2^d$ outgoing edges yet, we connect more edges from $(t-1,i)$ to a dummy node $(t,w+1)$. For each dummy node we connect $2^d$ edges to the dummy node in the next layers. It is easy to observe that $(M'^n)_{i,j}$ is exactly the probability that we start a random walk from $(0,i)$ and reach $(n,j)$. Now for every $i,j\in[w]$, define $B_{i,j}(x)$ to be the indicator for whether we will reach $(t,j)$ if we start from $(0,i)$ and follow $x\in(\bits{d})^n$. Then $\ex[x]{B_{i,j}(x)}=(M'^n)_{i.j}$. Take the given PRG $G$, we have 
$$\abs{\ex[r]{B_{i,j}(G(r))}-\ex[x]{B_{i,j}(x)}}\le\frac{\eps}{3}.$$
Now define the offline randomized algorithm $\widehat{\mathrm{Pow}}$ to be 
$$\widehat{\mathrm{Pow}}(M,y)_{i,j}=\ex[z]{B_{i,j}(G(\samp(y,z)))},$$
where $\samp$ is a $(\eps/3,\eps/w^2)$-sampler. By definition of sampler, with probability at least $(1-(\eps/w^2))$ over the choice of $y$, we have 
$$\abs{\widehat{\mathrm{Pow}}(M,y)_{i,j}-\ex[r]{B_{i,j}(G(r))}}\le \frac{\eps}{3}.$$ By union bound, with probability at least $(1-\eps)$, 
$$\nmax{\widehat{\mathrm{Pow}}(M,y)_{i,j}-\ex[r]{B_{i,j}(G(r))}}\le \frac{\eps}{3}$$
for every $i,j\in[w]$. Therefore by triangle inequality we have 
$$\nmax{\widehat{\mathrm{Pow}}(M,y)-M^n}\le \eps$$
with probability at least $1-\eps$.

Finally we compute the complexity of $\widehat{\mathrm{Pow}}$. By \Cref{lemma:samp-opt}, the required randomness in this offline randomized algorithm is $s+O(\log(1/\eps)+\log\log(w/\eps))$. The required processing space is the processing space for samplers and PRGs. Observe that the only sensitive data is the second input for sampler (i.e. $z$); the current node in the ROBP, which takes $\log(nw)$ bits to store; and a $d$-bit block in $G(\samp(y,z))$ indicating which entry of $M$ we should check. Therefore the required sensitive processing space is only $O(\log(nw/\eps))$ bits. 
\end{proof}

With Armoni's sampler trick, if we have any PRG for $(n,w+1,\log(3nw/\eps))$-ROBP, we can always plug it into the Saks-Zhou scheme regardless of whether it has the high-probability property. Specifically, as suggested in \cite{BCG18}, if we have a PRG of seed length $O(\log^{2}(n)+\log^{4/3}(w/\eps))$, we can even prove that $\BPL\subseteq \L^{4/3}$.

\subsection{Saks-Zhou-Armoni Scheme with PRPDs}
Finally we see how to apply a PRPD in the above scheme.
\begin{lemma}\label{lemma:Arm-general}
Suppose there exists an explicit $\poly(nw/\eps)$-bounded PRPD $(G,\rho)$ for $(n,w+1,\log(3nw/\eps))$-ROBP with error $\eps/3$ which has seed length $s$. Then there exists an offline randomized algorithm which approximates the $n$-th power of any substochastic matrix within error $\eps$ with probability at least $1-\eps$. Moreover, such algorithm requires $s+O(\log(w/\eps))$ bits of randomness, $O(s+O(\log(w/\eps)))$ bits of reusable processing space and $O(\log(nw/\eps))$ bits of sensitive processing space. 
\end{lemma}
\begin{proof}
The proof is basically the same as \Cref{lemma:Armoni}, with the following two difference.
\begin{itemize}
\item 
$\widehat{\mathrm{Pow}}(M,y)_{i,j}$ is defined as $\ex[z]{\rho(\samp(y,z))\cdot B_{i,j}( G(\samp(y,z)))}$ instead.
\item
If $(G,\rho)$ is $k$-bounded, then we will choose $\samp$ as a $(\eps/6k,\eps/w^2)$ sampler instead.
\end{itemize}
It's not hard to verify the correctness. (With \Cref{lemma:sampler-general} which shows that samplers can also be used for functions with output range $[-k,k]$.) The required sensitive processing space is increased to $O(\log(nw/\eps)+\log(k))$, which is still $O(\log(nw/\eps))$ if $k=\poly(nw/\eps)$.
\end{proof}
One may notice that there might have negative output in our new definition of $\widehat{\mathrm{Pow}}$. However, this is not a problem when applying Saks-Zhou argument because we only rely on the non-negativeness of matrices $\overline{M_i}$, which is independent of the approximation algorithm we use. With the above lemma we have the following corollary, which better motivates the problem of getting improved seed length for PRPDs:

\begin{corollary}
If there exists a $\poly(nw/\eps)$-bounded explicit PRPD for $(n,w,d)$-ROBP with error $\eps$ which has seed length $O(\log^{2}(n)+(\log(w/\eps)+d)^{4/3})$, then $\BPL\subseteq \L^{4/3}$.
\end{corollary}
\begin{proof}
Apply the Saks-Zhou scheme (\Cref{lemma:SZ-general}), and take $n_1=2^{\log^{2/3}(n)}$, $n_2=\log^{1/3}(n)$, $d=10\log(n)$ and $\eps=2^{-2d}$. The required subprocedure $\widehat{\mathrm{Pow}}$ would be approximating the $n_1$-th power of $n\times n$ substochastic matrices within error $\eps$. By \Cref{lemma:Arm-general}, there exists an offline randomized algorithm which approximates $M^{n_1}$ within error $\eps=2^{-2d}=\poly(1/n)$, which requires sensitive processing space $O(\log (n))$ and offline randomness + reusable processing space $O(\log^2(n_1)+\log^{4/3} n)=O(\log^{4/3}(n))$. Therefore the total space complexity is $O(\log(n)\cdot n_2+\log^{4/3}(n))=O(\log^{4/3}(n))$.
\end{proof}

\begin{remark}
Note that while we only construct PRPDs for $(n,w,1)$-ROBP in this paper, it is possible to adapt our construction to get PRPDs for $(n,w,d)$-ROBP with seed length $O(\log n\log(nw)\log\log (nw)+\log(1/\eps)+d)$: simply replace the base case with a sampler with $d$-bit output. Since our PRPD construction doesn't imply better derandomization for $\BPL$ anyway, we keep $d=1$ for simplicity.
\end{remark}

\section{Proof of \Cref{lemma:samp-opt}}\label{sec:proof-of-samp}
\begin{lemma}[\Cref{lemma:samp-opt}, restated. \cite{RVW01,Gol11}]
For every $\delta,\eps>0$ and integer $m$, there exists a $(\eps,\delta)$-sampler $f:\bits{n}\times\bits{d}\to\bits{m}$ s.t. $d=O(\log\log(1/\delta)+\log(1/\eps))$ and $n=m+O(\log(1/\delta))+O(\log(1/\eps))$. Moreover, for every $x,y$, $f(x,y)$ can be computed in space $O(m+\log(1/\delta)+\log(1/\eps))$.
\end{lemma}

We will use the equivalence between seeded randomness extractor and oblivious sampler by Zuckerman~\cite{Zuc97}. To achieve the parameter we need, we need a ``high-entropy seeded extractor" such that the seed length only depends on entropy loss but not the length of source. We will use the standard ``block-source" construction for high-entropy extractor which can be found in \cite{GW97,RVW01}. For simplicity, we will use simple composition instead of zig-zag composition~\cite{RVW01} because we are not aiming for optimal entropy loss. We will use the following standard lemmas for the extractor construction. Some of the following lemmas are implicit in their original source, and we recommend the readers to see \cite{GUV09,Vad12} for a proof.
\begin{definition}[\cite{CG88}]
$(X_1,X_2)$ is a \emph{$(k_1,k_2)$-block source} if $X_1$ is a $k_1$-source, and for every $x_1\in\Supp(X)$, $X_2$ conditioned on $X_1=x_1$  is a $k_2$-source.
\end{definition}
\begin{lemma}[\cite{GW97}]\label{lemma:block}
Let $X\in\bits{n}$ be a $(n-\Delta)$ source. Then for every integer $0\le t\le n $, $X$ is $\eps$-close to a $(t-\Delta,n-t-\Delta-\log(1/\eps))$-block source $(X_1,X_2)$ where $X_1\in\bits{t}$ and $X_2\in\bits{n-t}$.
\end{lemma}
\begin{lemma}[\cite{NZ96}]\label{lemma:compo}
Let $E_1:\bits{n_1}\times\bits{d}\to\bits{d_2}$ be a $(k_1,\eps_1)$ extractor and $E_2:\bits{n_2}\times\bits{d_2}\to\bits{m}$ be a $(k_2,\eps_2)$ extractor. Define $E((x_1,x_2),s)=E_1(x_2,E_2(x_1,s))$. Then for every $(k_1,k_2)$-block source $(X_1,X_2)\in\bits{n_1}\times\bits{n_2}$, $E((X_1,X_2),U_d)$ is ${(\eps_1+\eps_2)}$-close to uniform.
\end{lemma}
\begin{lemma}[\cite{GW97}]\label{lemma:GW}
For every $\eps,\Delta>0$ and integer $n$ there exists a $(n-\Delta,\eps)$ extractor $E:\bits{n}\times\bits{d}\to\bits{n}$ with $d=O(\Delta+\log(1/\eps))$, and for every $x,y$, $E(x,y)$ can be computed in space $O(n+\log(1/\eps))$.
\end{lemma}
\begin{lemma}[\cite{GUV09,KNW08}]\label{lemma:GUV}
For every $\eps>0$, integer $m>0$ and $n\ge 2m$, there exists a $(2m,\eps)$ extractor $E:\bits{n}\times\bits{d}\to\bits{m}$ with $d=O(\log m + \log(1/\eps))$, and for every $x,y$, $E(x,y)$  can be computed in space $O(m+\log(1/\eps))$.
\end{lemma}
\begin{lemma}[\cite{Zuc97}]\label{lemma:Zuc}
Every $(n-\log(1/\delta)-1,\eps)$-extractor is a $(\eps,\delta)$-sampler.
\end{lemma}
\noindent
Now we show how to construct the sampler we need, and that it is indeed space efficient.
\begin{proof}
Let $\Delta=\log(1/\delta)+1$. Let $E_1:\bits{m}\times\bits{d_1}\to\bits{m}$ be an $(m-\Delta,\eps/3)$-extractor from Lemma~\ref{lemma:GW}, w.l.o.g. assume that $d_1\ge \Delta+\log(3/\eps)$. Then let $E_2:\bits{3d_1}\times\bits{d}\to\bits{d_1}$ be an $(2d_1,\eps/3)$-extractor from Lemma~\ref{lemma:GUV}. Then we claim that $E:\bits{m+3d_1}\times\bits{d}\to\bits{m}$, defined as $E((x_1,x_2),s)=E_1(x_1,E_2(x_2,s))$, is a $(m+3d_1-\Delta,\eps)$ extractor, and hence a $(\eps,\delta)$ sampler by Lemma~\ref{lemma:Zuc}. 

To prove the claim, consider any $(m+3d_1-\Delta)$-source $X$. By Lemma~\ref{lemma:block}, $X$ is $(\eps/3)$-close to a $(m-\Delta,3d_1-\Delta-\log(3/\eps))$-block source $(X_1,X_2)\in\bits{3d_1}\times\bits{m}$. By Lemma~\ref{lemma:compo}, $E_1(X_1,E_2(X_2,U_d))$ is $2\eps/3$-close to uniform. Since $E(X,U_d)$ is $\eps/3$-close to $E_1(X_1,E_2(X_2,U_d))$, by triangle inequality it is $\eps$-close to uniform. Moreover, $d=O(\log (d_1/\eps))=O(\log\log(1/\delta)+\log(1/\eps))$, $n=m+3d_1=m+O(\log(1/\delta)+\log(1/\eps))$, and the required space to compute $E$ is $O(m+d_1+\log(1/\eps))=O(m+\log(1/\eps)+\log(1/\delta))$.
\end{proof}
\end{document}